\documentclass[10pt]{article}
\usepackage[T1]{fontenc}
\usepackage{lmodern}
							
\usepackage[frenchmath]{mathastext}
\usepackage{natbib}       									
\usepackage{amsmath}                        
\usepackage{graphicx} 										
\usepackage{caption}
\usepackage{booktabs}
\usepackage{subfigure}
\usepackage{float} 	
\usepackage{enumitem}												
\usepackage{mdwlist}												
\usepackage[dvipsnames]{xcolor}							
\usepackage[plainpages=false, pdfpagelabels]{hyperref} 
	\hypersetup{
		colorlinks   = true,
		citecolor    = RoyalBlue,
		linkcolor    = RubineRed, 
		urlcolor     = Turquoise
	}
\usepackage{amssymb}                        
\usepackage[mathscr]{eucal}                 
\usepackage{dsfont}									
\usepackage[top=1.0in, bottom=1.0in, left=0.8in, right=0.8in]{geometry} 
\usepackage{ulem}
\usepackage{mathtools}                      
\mathtoolsset{showonlyrefs=true}       
\usepackage{amsthm}                
	\allowdisplaybreaks                       
	\theoremstyle{plain}
	\newtheorem{theorem}{Theorem}
	\newtheorem{lemma}[theorem]{Lemma}       
	\newtheorem{proposition}[theorem]{Proposition}
	
	\theoremstyle{definition}
	\newtheorem{definition}{Definition}
	
	\newtheorem{remark}{Remark}

\newcommand\Eb{\mathds{E}}
\newcommand\Fb{\mathds{F}}

\newcommand\Pb{\mathds{P}}

\newcommand\Rb{\mathds{R}}

\newcommand\Nb{\mathds{N}}

\newcommand\Ac{\mathscr{A}}
\newcommand\Bc{\mathscr{B}}

\newcommand\Fc{\mathscr{F}}
\newcommand\Gc{\mathscr{G}}
\newcommand\Hc{\mathscr{H}}

\newcommand\Oc{\mathscr{O}}
\newcommand\Pc{\mathscr{P}}

\newcommand\Xc{\mathscr{X}}
\newcommand\Yc{\mathscr{Y}}

\newcommand\eps{\varepsilon}

\newcommand\sig{\sigma}

\newcommand\Lam{\Lambda}
\newcommand\gam{\gamma}

\newcommand\lam{\lambda}

\newcommand\Pbh{\widehat{\Pb}}
\newcommand\Ebh{\widehat{\Eb}}
\newcommand\Ach{\widehat{\Ac}}
\newcommand\Gch{\widehat{\Gc}}
\newcommand\Xch{\widehat{\Xc}}
\newcommand\Ych{\widehat{\Yc}}

\newcommand\Bh{\widehat{B}}

\newcommand\Gamh{\widehat{\Gamma}}
\newcommand\Pch{\widehat{\Pc}}

\newcommand\Pbt{\widetilde{\Pb}}
\newcommand\Act{\widetilde{\Ac}}

\newcommand\Gct{\widetilde{\Gc}}
\newcommand\Xct{\widetilde{\Xc}}
\newcommand\Yct{\widetilde{\Yc}}

\newcommand\Gamt{\widetilde{\Gamma}}
\newcommand\Pct{\widetilde{\Pc}}

\newcommand{\dd}{\mathrm{d}}
\newcommand{\ee}{\mathrm{e}}

\newcommand{\eqlnostar}[2]{\begin{align}\label{#1}#2\end{align}}
\newcommand{\eqstar}[1]{\begin{align*}#1\end{align*}}

\makeatother

\begin{document}

\title{The implied Sharpe ratio}

\author{
Ankush Agarwal
\thanks{Adam Smith Business School, University of Glasgow. \textbf{e-mail}: \url{ankush.agarwal@glasgow.ac.uk}}
\and
Matthew Lorig
\thanks{Department of Applied Mathematics, University of Washington. \textbf{e-mail}: \url{mlorig@uw.edu}}
}

\date{This version: \today}

\maketitle

\begin{abstract}
In an incomplete market, including liquidly-traded European options in an investment portfolio could potentially improve the expected terminal utility for a risk-averse investor. However, unlike the \textit{Sharpe ratio}, which provides a concise measure of the relative investment attractiveness of different underlying risky assets, there is no such measure available to help investors choose among the different European options.  We introduce a new concept -- the \textit{implied Sharpe ratio} -- which allows investors to make such a comparison in an incomplete financial market. Specifically, when comparing various European options, it is the option with the highest implied Sharpe ratio that, if included in an investor's portfolio, will improve his expected utility the most. Through the method of Taylor series expansion of the state-dependent coefficients in a nonlinear partial differential equation, we also establish the behaviour of the implied Sharpe ratio with respect to an investor's risk-aversion parameter. 
In a series of numerical studies, we compare the investment attractiveness of different European options by studying their implied Sharpe ratio.

\vspace{0.5pc}
\end{abstract}
\section{Introduction}
\label{sec:intro}
The \textit{Sharpe ratio}, defined as the ratio of a portfolio's expected return to its standard deviation, was introduced by \cite{sharpe1966mutual} as a concise measure of mutual fund performance. Several other works have since used the idea under different settings to measure the performance of investment portfolios composed of risky assets. \cite{jensen1969risk} considered the effect of differential risk on required returns of different risky assets. Using a capital asset pricing model, he showed that the investment portfolios can be compared by looking at the difference between the actual returns on a portfolio and the expected returns on that portfolio conditional on its level of systematic risk and actual returns of the market portfolio. Another approach towards portfolio selection was proposed by \cite{merton1969lifetime} in which he incorporated the risk-preference of the investor through the risk-aversion parameter of a utility function.  He considered an investor that trades risky assets in the market to maximise his expected terminal utility. Under the assumption that the log-asset returns are normally distributed with a constant Sharpe ratio, he showed that, assuming a constant relative risk-aversion utility function, the investor derives a higher expected utility from assets with higher Sharpe ratio. 

\cite{hodges1989optimal} used the idea of expected utility maximisation to create optimal replication portfolios for contingent claims in a market with transaction costs. Their work led to the idea of \textit{indifference pricing} of contingent claims in \textit{incomplete markets} (see \cite{carmona2008indifference} for a collection of the related works). The framework of \textit{indifference pricing} allows to identify an economically justifiable price of a contingent claim in an {incomplete market}. In such markets, it is not possible to perfectly replicate any contingent claim using a self-financing portfolio of the underlying risky assets and determine a unique price (see, for example, Section 10.3 \cite{delbaen2006mathematics} for a complete mathematical characterisation of incomplete markets). Typically, the Hamilton-Jacobi-Bellman (HJB) equation satisfied by the indifference price is solved numerically. Recently, in an incomplete market framework based on a class of stochastic volatility models, \cite{lorig-4} approximately solved the HJB equation for the indifference price of a European option. He developed approximation techniques to solve nonlinear partial differential equations with state-dependent coefficients using a Taylor series expansion method developed in a series of papers: \cite{pagliarani2012analytical}, \cite{lorig2015analytical}, \cite{lorig-pagliarani-pascucci-2}. The corresponding approximate implied volatilities of the indifference prices are also provided by \cite{lorig-4}.

As the underlying risky assets cannot be used to perfectly hedge a European option in an incomplete market, an investor looking to create an investment portfolio can benefit by including a European option in it. This raises an important question of how to measure the investment attractiveness of different European options available in the market. Similar to how the Sharpe ratio can be used to measure the performance of an underlying risky asset, in this paper we introduce the concept of an \textit{implied Sharpe ratio} which can be viewed as a measure of a European option's worth to an investor.  Typically, the value of a European option is quantified in terms of its implied volatility.  When comparing two options, it is the option with a higher implied volatility that is considered to be the more expensive option. However, an option's implied volatility does not capture its worth to an investor. Relative to not owning the option, buying or selling an option at a given implied volatility may either increase or decrease an investor's expected terminal utility. The implied Sharpe ratio, as defined in this work, aims to address the issue of measuring an option's worth to an investor.

In the investment literature, the behavioural preference of an investor is captured through the risk-aversion parameter of his utility function.  Many classical studies of investment under uncertainty derive comparative statics with respect to the risk-aversion parameter. For example, it is well established that a risk-averse investor is willing to pay a higher premium to insure himself against risk. \cite{jewitt1987risk} extended this classical comparative static of risk-aversion parameter in the presence of an extra source of uncertainty. \cite{eeckhoudt1995risk} studied the impact of cost and price changes on the inventory of a risk-averse newsboy who must decide to purchase newspapers in order to sell them later. Similarly in our work, we aim to express the implied Sharpe ratio in terms of the model parameters and an investor's risk-aversion parameter in order to help him choose between European options with different strikes and times-to-maturity. 

The rest of the paper is organised as follows: In Section \ref{sec:implied-sharpe}, we define the implied Sharpe ratio and prove its uniqueness and existence in a general market setting.  In Section \ref{sec:market}, we introduce a general local stochastic volatility model which is popular among the practitioners and derive the equation for the implied Sharpe ratio in terms of the model parameters and risk-aversion parameter. In Section \ref{sec:asymptotics}, we use a Taylor series expansion technique to develop semi-explicit approximations of the implied Sharpe ratio.  In Section \ref{sec:examples} we study the particular forms of implied Sharpe ratio using different examples of local stochastic volatility models and illustrate {its importance in identifying the most suitable European option for investment among the available choices.} Mathematical proofs and a collection of formulas for calculating the implied Sharpe ratio under different models are presented in Appendices \ref{sec:proofs}--\ref{sec:formulas Heston}.

\section{The implied Sharpe ratio}
\label{sec:implied-sharpe}
Let us consider a frictionless financial market consisting of one risk-less asset with constant price and one non-dividend paying risky asset represented by a real-valued semi-martingale $S$ defined on a filtered probability space $(\Omega, \Fc,(\Fc_t), \Pb)$ satisfying the usual conditions. Let $W = (W_t)_{0 \leq t \leq T}$ denote the wealth process of an investor who invests in $\pi_t$ number of shares of $S$ at time $t$ starting from initial wealth level $w \in \Rb$. Then, the wealth process $W$ satisfies
\eqlnostar{eq:wealth}{
W_t = w + \int^t_0 \pi_u \dd S_u,
}
where $\pi = (\pi_t)_{0 \leq t \leq T}$ is a real-valued predictable process such that the stochastic integral above is well defined. In addition, we assume that the investor also owns European-style contingent claims on the risky asset. Over a fixed time horizon $T> 0,$ the investor trades the risky asset in order to maximise his expected terminal utility. We suppose the investor's utility function $U: \Rb \to \Rb$ is strictly concave, strictly increasing and belongs to the class $C^{\infty}(\Rb)$ of infinitely differentiable functions. The value of the investment to investor is captured by a function defined as follows:
\begin{definition}
The \textit{value function} $V$ of an investor with wealth process \eqref{eq:wealth} and who owns $\nu$ European-style contingent claims, each with payoff function $\varphi(S_T),$ is defined as
\eqstar{
V(t,s,w,\nu) := \sup_{\pi \in \Pi} \Eb_{t,s,w} \Bigl[ U \bigl(W_T + \nu \varphi(S_T) \bigr) \Bigr], \quad 0 \leq t < T, s > 0, w, \nu \in \Rb,
}
where $\Pi$ is the set of {admissible strategies} given as 
\eqstar{
\Pi := \Bigl\{ \pi:   (\pi_t)_{0 \leq t \leq T} \text{ is predictable and } \Eb_{0,s,w}\int^T_0 \pi_t^2  \dd \langle S, S\rangle_t < \infty\Bigr\}.
}
\end{definition}
When $\nu < 0$ the investor is a seller of the contingent claim, and when $\nu > 0$ he is a buyer. If for a given European option payoff we can explicitly solve for $V,$ we can find out whether the investor will prefer to sell or buy that European option. In an analogy to the Black-Scholes price of a European option, we define the \textit{Merton value function} as the value function of an investor when we assume that the risky asset price follows a geometric Brownian motion.
\begin{definition}
Suppose the risky asset dynamics is given by the Black-Scholes model in which $\dd S_t = \mu \, S_t \dd t + \sig\, S_t \dd B_t$, and where, $\mu \in \Rb$ is the expected rate of return, $\sigma \in \Rb_+ $ is the volatility and $B = (B_t)_{t \geq 0}$ is a standard Brownian motion. The \textit{Merton value function} $V^{\text{M}}$ of an investor with wealth process \eqref{eq:wealth} is then defined as 
\eqstar{
V^{\text{M}}(t,w;\lam) := \underset{\pi \in \Pi} {\sup}\Eb_{t,s,w} [U(W_T)],
}
where $\lambda := \tfrac{\mu}{\sigma}$ is the Sharpe ratio.
\end{definition}
The Merton value function is an increasing function with respect to $\lambda.$ Thus, a risky asset with higher Sharpe ratio will lead to a higher Merton value function. Just as the implied volatility is used to link the market-observed European option price to the Black-Scholes model's European option price, we define the \textit{implied Sharpe ratio} to link an investor's value function in a general risky-asset price model to the Merton value function in the Black-Scholes model.  
\begin{definition}
\label{def:imp Sharpe ratio}
Suppose an investor with wealth process \eqref{eq:wealth} owns $\nu$ European-style contingent claims with unit price $p.$ Further suppose that $V(t,s,w-\nu p, \nu) \geq U(w)$ for all $t \in [0,T]$ and any $s > 0, w, \nu \in \Rb.$ Then, the {\it implied Sharpe ratio} is the {unique} positive solution $\Lam$ of the equation
\eqlnostar{eq:impSharpe}{
V^{\text{M}}(t,w;\Lam) = V(t,s,w - \nu p,\nu),
}
where $V^{\text{M}}$ is the Merton value function and $V$ is the investor's value function. 
\end{definition}
%The above definition says that 
According to the above definition, if an investor assumes that the risky asset price follows a geometric Brownian motion with Sharpe ratio $\Lambda$  and invests only in that risky asset, he will obtain the same expected utility as he would by investing in a portfolio composed of the risky asset and a European option, in a general risky-asset price model. European calls and puts are usually compared based on their implied volatilities. When plugged in the Black-Scholes model's European call/put formula, implied volatility provides the market price of the option. However, it does not provide any measure of the option's worth to an investor who is looking to maximize his expected terminal utility. The implied Sharpe ratio rectifies this shortcoming of implied volatility. 

Through the implied Sharpe ratio, we can connect an investor's value function $V$ in an incomplete market to the classical Merton value function $V^{\text{M}}.$ In doing so, our aim is two-fold: First, by computing the implied Sharpe ratio for $\nu > 0$ (buying the option) or $\nu < 0$ (selling the option), we can compare it to the case $\nu = 0$ (no position in the option). From Definition \ref{def:imp Sharpe ratio}, it is clear that a higher implied Sharpe ratio, compared to the other possibilities, will deliver a higher value to an investor who is looking to maximise his expected terminal utility over a finite time horizon. Thus, the implied Sharpe ratio can tell us whether buying ($\nu > 0$) or selling ($\nu < 0$) a contingent claim can improve the investor's utility as opposed to just investing in the risky asset ($\nu =0$) underlying the option. This comparison can also be made directly by comparing the value functions corresponding to the different possibilities ($\nu <0, \nu >0,\nu =0$). However, as we will see later, in the case of an \textit{exponential utility function}, the implied Sharpe ratio is independent of the starting wealth level. Thus, it provides a standard measure of investment worth as opposed to the value function which is dependent on the starting level of wealth.

Second, by computing the implied Sharpe ratio for two European options with different strikes and maturities, we can compare their relative worth to an investor. The European option with the higher implied Sharpe ratio will be more attractive to the investor. Such a comparison can be made across the whole range of available European options, thus making the implied Sharpe ratio a better measure than the implied volatility to assess the worth of an option for investment. {The implied volatility is a unitless quantity which provides a better comparison between different European options than their respective prices. In the same spirit, the implied Sharpe ratio allows a more intuitive comparison between the different European options than their respective expected utility values which is dependent on initial wealth.} As mentioned earlier, in the case of an \textit{exponential utility function}, the implied Sharpe ratio will be independent of the starting wealth level, thus making it a better choice than directly observing the value function, which depends on the starting wealth. Investors with different wealth levels can just use the implied Sharpe ratio to make informed investment decisions instead of using the value function. Next, we establish the theoretical result which ensures the existence and uniqueness of the implied Sharpe ratio. 

\begin{theorem}
The implied Sharpe ratio defined in \eqref{eq:impSharpe} exists and is unique.
\end{theorem}
\begin{proof}
The function $V^{\text{M}}(t,w;\lambda)$ is the solution of the following partial differential equation:
\eqlnostar{eq:value function}{
\partial_t V - \frac{1}{2} \lam^2 \frac{(\partial_w V)^2}{\partial^2_wV} = 0, \qquad V(T,w) = U(w).
}
It is also strictly increasing and strictly concave with respect to the state variable $w$. Thus, we have $\partial^2_w V^{\text{M}}< 0.$ Next, suppose $\lam> \lam^\prime,$ and let $V^{\text{M}}(t,w;\lam)$ and $V^{\text{M}}(t,w;\lam^\prime)$ denote the solutions to Equation \eqref{eq:value function} for $\lambda$ and $\lambda^\prime,$ respectively. Then, we get
\eqstar{
&\partial_t V^{\text{M}}(t,w;\lam^\prime) - \frac{1}{2} \lam^2 \left(\frac{\partial_w V^{\text{M}}}{\partial^2_w V^{\text{M}}}\right)^2 (t,w;\lam^\prime) \partial^2_w V^{\text{M}} (t,w;\lam^\prime)\\
&= \partial_t V^{\text{M}}(t,w;\lam^\prime) - \frac{1}{2} \lam^\prime \left(\frac{\partial_w V^{\text{M}}}{\partial^2_w V^{\text{M}}}\right)^2 (t,w;\lam^\prime) \partial^2_w V^{\text{M}} (t,w;\lam^\prime)\\
&+ \frac{1}{2}\bigl({\lam^\prime}^2 - \lam^2 \bigr) \left(\frac{\partial_w V^{\text{M}}}{\partial^2_wV^{\text{M}}}\right)^2 (t,w;\lam^\prime) \partial^2_w V^{\text{M}} (t,w;\lam^\prime) > 0.
}
Therefore, $V^{\text{M}}(t,w;\lam^\prime)$ is a subsolution to the equation solved by $V^{\text{M}}(t,w;\lam)$ with the terminal condition which is independent of $\lam$ and $V^{\text{M}}(T,w;\lam^\prime) = V^{\text{M}}(T,w;\lam).$ Thus, we have that for $\lam > \lam^\prime$ and $t < T,$ $V^{\text{M}}(t,w;\lam) > V^{\text{M}}(t,w;\lam^\prime).$ In other words, $V^{\text{M}}$ is strictly increasing in $\lam.$ Moreover, for any $\lam, V^{\text{M}}(t,w;\lam) \geq U(w)$ which guarantees the existence and uniqueness of the solution to \eqref{eq:impSharpe}.  
\end{proof}

%For a risk-averse investor, finding the appropriate option to invest is crucial and therefore, the implied Sharpe ratio is an important measure to assess its worthiness. In a complete market, by trading in the underlying asset, any European option payoff can be replicated which makes a European option redundant for investment purposes. 
In the next section, we study the implied Sharpe ratio in a class of incomplete market models that are Markov. In the absence of closed-form formulas of the implied Sharpe ratio, we use a Taylor series expansion technique for nolinear PDEs developed by \cite{lorig-4} to find semi-explicit approximation in terms of the risk-aversion parameter and other model parameters.

\section{Markov market setting}
\label{sec:market}
To study the behaviour of the implied Sharpe ratio with respect to market parameters, we specialize to the setting of a Markov market model. We suppose that the dynamics of $S$ are of the following form
\eqlnostar{eq:main sde1}{
S_t 	&=\exp \bigl( X_t \bigr), \nonumber\\
\dd  X_t 	&=	\Bigl( \mu(X_t,Y_t) - \tfrac{1}{2} \sigma^2(X_t,Y_t) \Bigr)  \dd t + \sigma(X_t,Y_t)  \dd B_t^X ,\\
\label{eq:main sde2}
\dd Y_t &=	c(X_t,Y_t) \dd t + \beta(X_t,Y_t) \Bigl( \rho  \dd B_t^X +   \sqrt{1 - \rho^2} \dd B_t^Y \Bigr) , 
}
where $B^X = (B_t^X)_{0 \leq t \leq T}$ and $B^Y=(B_t^Y)_{0 \leq t \leq T}$ are independent Brownian motions. We assume that the system of stochastic differential equations \eqref{eq:main sde1}-\eqref{eq:main sde2} admits a unique strong solution $(X,Y)$ adapted to filtration $\Fb = (\Fc_t)_{0 \leq t \leq T}.$ Next, for $\pi_t$ units of currency invested in $S$ at time $t,$ the investor's wealth process $W$ satisfies the following equation
\eqstar{
\dd W_t = \frac{\pi_t}{S_t} \dd S_t = \pi_t \mu(X_t,Y_t) \dd t + \pi_t \sigma(X_t,Y_t) \dd B^X_t.
}
The value function of an investor with  $\nu$ European-style options and initial wealth level $w$ is given as 
\eqlnostar{eq:value function def}{
V(t,x,y,w,\nu) = \underset{\pi \in \Pi} {\sup}\Eb_{t,x,y,w} [ U (W_T + \nu\varphi(X_T))].
}
Each European option with payoff function $\varphi$ and maturity $T$ is assumed to have price $p$, which is obtained by computing the expectation of the option's payoff under the market's chosen pricing measure.
We adopt the primal approach to solve for $V$ and assume that the value function belongs to $C^{1,2,2,2}([0,T] \times \Rb^3).$ Under this assumption, by following the usual dynamic programming principle (see, for example, Chapter 3 \cite{pham2009continuous}), $V$ satisfies the Hamilton-Jacobi-Bellman (HJB) equation 
\eqlnostar{eq:HJB main value func}{
0= (\partial_t + \Ac) V + \underset{\pi \in R}{\max} \, \Ac^\pi V,&  &V(T,x,y,w,\nu) = U(w + \nu \varphi(x)) ,
}
where the operators $\Ac$ and $\Ac^\pi$ are given as 
\eqstar{
\Ac &=\bigl (\mu(x,y) - \frac{1}{2}\sigma^2(x,y)\bigr) \partial_x +  \frac{1}{2}\sigma^2(x,y) \partial^2_x  + c(x,y)  \partial_y +  \frac{1}{2}\beta^2(x,y)\partial^2_y + \rho \sigma(x,y) \beta(x,y) \partial_x \partial_y,\\
\Ac^\pi &= \pi(t,x,y,w) \mu(x,y) \partial_w + \frac{1}{2}\pi^2(t,x,y,w)\sigma^2(x,y)\partial^2_w +  \pi(t,x,y,w) \rho \sigma(x,y) \beta(x,y) \partial_y \partial_w\\
& \quad +  \pi(t,x,y,w) \sigma^2(x,y) \partial_x \partial_w.
}
The candidate optimal strategy $\pi^\star$ is obtained by maximizing $\Ac^\pi V$ which gives us
\eqlnostar{eq:opt strat}{
\pi^\star = \underset{\pi \in \Rb}{\arg\max} \, \Ac^\pi V = - \left(\frac{\mu (\partial_w V) + \rho \beta \sigma (\partial_y \partial_w V) +\sigma^2 (\partial_x \partial_w V)}{\sigma^2 \partial^2_w V} \right).
}
In the above formula, we have suppressed the arguments $(t,x,y,w,\nu)$ for simplicity of notation. We will do so from now on, wherever it causes no confusion, in order to keep the notation manageable. 

Inserting the optimal strategy $\pi^\star$ in HJB equation \eqref{eq:HJB main value func} yields the following:
\eqstar{
0 = (\partial_t + \Ac) V + \Hc(V),& &V(T,x,y,w,\nu) = U(w  + \nu \varphi(x)),
}
where the \textit{Hamiltonian} $\Hc(V)$ is a nonlinear term given as 
\eqlnostar{eq:value func HJB}{
\Hc(V) = &-\frac{1}{2}\lam^2 \frac{(\partial_w V)^2}{\partial^2_w V} - \rho \beta \lam \frac{(\partial_w V)(\partial_y \partial_w V)}{\partial^2_w V} - \frac{1}{2}\rho^2 \beta^2  \frac{(\partial_y \partial_w V)^2}{\partial^2_w V}\\
&-\mu \frac{(\partial_w V)(\partial_x \partial_w V)}{\partial^2_w V} - \rho \sigma \beta \frac{(\partial_x \partial_w V)(\partial_y \partial_w V)}{\partial^2_w V} - \frac{1}{2}\sig^2  \frac{(\partial_x \partial_w V)^2}{\partial^2_w V}, \quad \lam := \frac{\mu}{\sig}.
}
Note that $\lambda$ defined above is the \textit{instantaneous Sharpe ratio} of the risky asset $S$. In the rest of the study, we fix the utility function $U$ to be of the exponential form:
\eqstar{
U(x) = -\frac{1}{\gamma} \ee^{-\gamma x}, \quad \gamma >0,
}
where $\gamma$ is the risk-aversion parameter. This choice of utility function allows us to obtain a form of the implied Sharpe ratio which is independent of the initial wealth level. The Merton value function in this case is given as
\eqstar{
V^{\text{M}}(t,w;\lam) = -\frac{1}{\gam} \exp \left( -\gam w - (T -t) \tfrac{1}{2}\lam^2\right).
}
One can verify by direct substitution that the above expression satisfies \eqref{eq:value function}.
Inspired by the form of the Merton value function for exponential utility, we make the following ansatz to solve the HJB equation \eqref{eq:value func HJB}:
\eqlnostar{eq:ansatz1}{
V(t,x,y,w,\nu) = -\frac{1}{\gam} \exp\left( - \gamma w + \psi(t,x,y,\nu)\right).
}
We find that the function $\psi$ satisfies the following equation
\eqlnostar{eq:main psi}{
0 = \left( \partial_t + \Act \right) \psi+ \Bc(\psi),& &\psi(T,x,y,\nu)= -\gamma \nu \varphi(x),
}
where the linear operator $\Act$ and the nonlinear operator $\Bc$ are given by 
\eqstar{
\Act = \tfrac{1}{2}\sig^2 (\partial^2_x - \partial_x) + (c - \rho \beta \lam) \partial_y + \tfrac{1}{2} \beta^2 \partial^2_y + \rho \sig \beta \partial_x \partial_y, & &\Bc(\psi) = (1-\rho^2) (\tfrac{1}{2}\beta^2) (\partial_y \psi)^2 - \tfrac{1}{2} \lam^2.
}
To find a formula for the implied Sharpe ratio $\Lam,$ we solve the following equation
\eqlnostar{eq:Merton modified value}{
V^{\text{M}}(t,w;\Lam) = V(t,x,y,w-\nu p,\nu).
}
It can be checked that \eqref{eq:Merton modified value} is satisfied if and only if 
\eqlnostar{eq:formula1}{
\gamma \nu p + \psi = -(T-t) \tfrac{1}{2} (\Lam)^2 ,
}
where, as a reminder, $p$ is the price of the European option with payoff $\varphi(X_T)$, as computed under the market's chosen pricing measure. As $p$ is independent of $\gamma,$ we can obtain a relationship between the implied Sharpe ratio and investor's risk-aversion parameter by solving for $\psi.$ In equation \eqref{eq:main psi}, we observe that $\psi$ depends on $\gamma$ only through the terminal condition. However, due to the presence of non-linearity $\Bc(\psi)$ in \eqref{eq:main psi}, deriving a closed-form formula for the relationship between $\psi$ and $\gamma$ is not possible. Even without the non-linearity in \eqref{eq:main psi}, the standard theory of partial differential equations (see Chapter 2, Section 4 \cite{friedman2008partial}) is not directly applicable to derive the relationship. 

\section{Asymptotic approximation formulas}
\label{sec:asymptotics}
In this section, we will use a Taylor series expansion method, as developed by \cite{lorig-4}, to derive asymptotic approximations for $\psi$, the option price $p$ and the implied Sharpe ratio $\Lam$.
\subsection{Asymptotic approximation formulas for $\psi$}
\label{sec:asymptotics psi}
Let $\chi$ be any of the coefficient functions appearing in the operators $\Act$ or $\Bc(\cdot),$ that is 
\eqstar{
\chi \in \{(\tfrac{1}{2} \sigma^2), (c-\rho \beta \lam), (\rho \sigma \beta),  (\tfrac{1}{2} \beta^2),(\tfrac{1}{2} \lam^2) \}.
} 
Next, fix a point $(\bar{x},\bar{y}) \in \Rb^2$ and define the following family of functions indexed by $\eps  \in  [0,1]:$
\eqlnostar{eq:eps coeff}{
\chi^\eps(x,y) &:= \chi(\bar{x} + \eps(x-\bar{x}),\bar{y} + \eps(y-\bar{y})).
}
Observe that
\eqstar{
\quad \chi^\eps(x,y)\big|_{\eps=1} = \chi(x,y), \qquad \chi^\eps(x,y)\big|_{\eps=0} = \chi(\bar{x},\bar{y}).
}
%Note above that $\chi^\eps(x,y)$ evaluated at $\eps=1$ gives the original function. 
Suppose the functions are analytic in a neighbourhood of $(\bar{x},\bar{y})$ so that we have
\eqstar{
\chi^\eps(x,y) &:= \sum^\infty_{n=0} \eps^n \chi_n(x,y),\\
\chi_n(x,y) &:=\sum^n_{k=0} \chi_{n-k,k}  \cdot (x - \bar{x})^{n-k}(y - \bar{y})^{k},\quad \chi_{n-k,k} := \frac{1}{(n-k)!k!} \partial^{n-k}_x\partial^{k}_y \chi(\bar{x},\bar{y}).
 }
For the final $m$-th order approximation formula, we will only need the functions to be $m$-times differentiable. However, it will simplify the presentation if, for now, we assume the stronger condition of analyticity. Using the coefficient functions indexed by $\eps$, consider the following family of PDEs indexed by $\eps:$
\eqlnostar{eq:pde psi}{
0 = \left( \partial_t + \Act^\eps \right) \psi^\eps + \Bc^\eps(\psi^\eps),&  &\psi^\eps(T,x,y,\nu) = -\gamma \nu \varphi(x),
}
where $\Act^\eps$ and $\Bc^\eps(\cdot)$ are obtained from $\Act$ and $\Bc(\cdot)$ and by replacing the coefficients in these operators with their $\eps$-counterparts defined in \eqref{eq:eps coeff} 
\eqstar{
\Act^\eps &= (\tfrac{1}{2}\sig^2)^\eps (\partial^2_x - \partial_x) + (c - \rho \beta \lam)^\eps \partial_y + (\tfrac{1}{2} \beta^2)^\eps \partial^2_y + (\rho \sig \beta)^\eps \partial_x \partial_y,\\
\Bc^\eps(\psi^\eps) &= (1-\rho^2) (\tfrac{1}{2}\beta^2)^\eps (\partial_y \psi^\eps)^2 - (\tfrac{1}{2} \lam^2)^\eps.
}
Next, we suppose that $\psi^\eps$ is given in powers of $\eps$ as
\eqlnostar{eq:approx psi}{
\psi^\eps &= \sum^\infty_{i=0} \eps^i \psi_i.
}
Once we are able to solve for each $\psi_i,$ the asymptotic approximation of $\psi$ is obtained by setting $\eps = 1$ in \eqref{eq:approx psi}. To obtain the respective order terms $\psi_i,$ we insert the expansion of $\psi^\eps$ into \eqref{eq:pde psi} and collect the terms of like powers of $\eps$. At the lowest order of $\eps,$ we have the following
\eqstar{
\Oc(1):& &0 = (\partial_t  + \Act_0) \psi_0 + (1-\rho^2) (\tfrac{1}{2}\beta^2)_0 (\partial_y \psi_0)^2 -  (\tfrac{1}{2}\lam^2)_0, & &\psi_0(T,x,y,\nu) =-\gamma \nu \varphi(x),
}
where we have defined 
\eqlnostar{eq:n order A}{
\Act_n := (\tfrac{1}{2}\sig^2)_n (\partial^2_x - \partial_x) + (c - \rho \beta \lam)_n \partial_y + (\tfrac{1}{2} \beta^2)_n \partial^2_y + (\rho \sig \beta)_n \partial_x \partial_y, \quad n \in \{0\} \cup \Nb.
}
In the above, $\Act_0$ is a \textit{constant coefficient} differential operator since we have that $\chi_0:= \chi(\bar{x},\bar{y}).$ It can be seen that  $\psi_0$ is only a function of $(t,x).$ As such the equation becomes
\eqlnostar{eq:solve psi0}{
\Oc(1):& &0 = (\partial_t  + \Act_0) \psi_0  -  (\tfrac{1}{2}\lam^2)_0, & &\psi_0(T,x,\nu) = - \gamma \nu\varphi(x).
}
The higher order terms $m \geq 1$ are given as 
\eqlnostar{eq:solve psim}{
\Oc(\eps^m):& &0 = (\partial_t  + \Act_0) \psi_m + H_m, & &\psi_m(T,x,y,\nu)  = 0,
}
where the source terms $H_m$ are given as 
\eqstar{
\Oc(\eps^m):&  &H_m &= \sum^m_{k=1} \Act_k \psi_{m-k} - (\tfrac{1}{2}\lam^2)_m + (1-\rho^2)\sum_{k,i,j \in K_m} (\tfrac{1}{2}\beta^2)_k (\partial_y \psi_i) (\partial_y \psi_j),\\
& &K_m &= \{(i,k,j) \in \Nb^3_0 : i + j + k = m, \text{ and } i,j,k \neq m\}.
}
In particular, the first and second order source terms are given by 
\eqlnostar{eq:source H1}{
\Oc(\eps):& &H_1 &= \Act_1 \psi_0 -  (\tfrac{1}{2}\lam^2)_1, \\
\label{eq:source H2}
\Oc(\eps^2):& &H_2&= \Act_2 \psi_0 + \Act_1 \psi_1 -  (\tfrac{1}{2}\lam^2)_2 + (1-\rho^2) (\tfrac{1}{2}\beta^2)_0 (\partial_y \psi_1)^2.
}
%Note that to obtain an asymptotic approximation for $\Lambda$ in \eqref{eq:formula1}, we also require an asymptotic approximation formula for the European option price $p.$ It can be obtained using the same approach as for $\psi$. 

Before, we proceed to solve for terms $\psi_i,$ we state a few fundamental results which will be useful for our calculations. The constant coefficient elliptic operator $\Act_0$ in \eqref{eq:solve psi0}, gives rise to the semigroup $\Pct_0(t,t_1)$ defined by 
\eqlnostar{eq:semigroup P0}{
\Pct_0(t,t_1) \eta(x,y) := \int_{\Rb^2} \dd x \dd y \Gamt_0 (t,x,y;t_1,x_1,y_1) \eta(x_1,y_1),
} 
where $\eta: \Rb^2 \to \Rb^2$ is a generic test function and $\Gamt_0$ is the \textit{fundamental solution} to the linear operator $(\partial_t + \Act_0)$ given as 
\eqstar{
\Gamt_0(t,x,y;t_1,x_1,y_1) = \frac{1}{\sqrt{(2\pi)^2 |\widetilde{\mathbf{C}}|}} \exp \left(-\frac{1}{2}\widetilde{\mathbf{m}}^\top  \widetilde{\mathbf{C}}^{-1} \widetilde{\mathbf{m}}\right).
}
The covariance matrix $\widetilde{\mathbf{C}}$ and vector $\widetilde{\mathbf{m}}$ are given as 
\eqstar{
\widetilde{\mathbf{C}} = (t_1 - t) \begin{pmatrix}
(\sigma^2)_0 & (\rho \sigma \beta)_0\\
(\rho \sigma \beta)_0 &(\beta^2)_0 \\
\end{pmatrix}, \quad \widetilde{\mathbf{m}} = (t_1 - t) \begin{pmatrix}
x_1 - x - (t_1 - t)(-\tfrac{1}{2}\sigma^2)_0 \\
y_1 - y - (t_1 - t)(c - \rho \beta \lambda)_0 \\
\end{pmatrix}.
}
The semigroup $\Pct_0(t,t_1)$ satisfies the following property:
\eqstar{
\Pct_0(t,t_1)\Pct_0(t_1,t_2) = \Pct_0(t,t_2), \quad t \leq t_1 \leq t_2.
}
By Duhamel's principle, the unique classical solution (if it exists) to any PDE of the form:
\eqstar{
0 = (\partial_t + \Act_0) u + h,& &u(T,x,y) = g(x,y),
}
is given as 
\eqstar{
u(t) = \Pct_0(t,T) \, g + \int^T_t\dd t_1 \Pct_0(t,t_1) h(t_1).
}
In the above we have omitted the arguments $(x,y)$ for simplicity of notation. To obtain the terms in the expansion of $\psi^\eps$ and $p^\eps,$ it will also be helpful to introduce the following operators
\eqlnostar{eq:op X}{
&\Xct(t,t_1) := x + (t_1-t) \left(-(\tfrac{1}{2} \sigma^2)_0 + 2 (\tfrac{1}{2} \sigma^2)_0 \partial_x + (\rho \sigma \beta)_0 \partial_y \right),&  &t_1 \geq t,\\
\label{eq:op Y}
&\Yct(t,t_1) := y + (t_1-t) \left( (c - \rho \beta \lambda)_0 +  2 (\tfrac{1}{2} \beta^2)_0 \partial_y + (\rho \sigma \beta)_0 \partial_x \right),& &t_1 \geq t.
}
By direct computations, we can check that the operators $\Xct(t,t_1)$ and $\Yct(t,t_1)$ commute and have the following property
\eqlnostar{eq:prop commute}{
\bigl(\Xct(t,t_1)\bigr)^n \bigl(\Yct(t,t_1)\bigr)^m \Gamt_0 (t,x,y;t_1,x_1,y_1) = x_1^n \, y_1^m \, \Gamt_0(t,x,y;t_1,x_1,y_1), \quad n,m \in \{0\} \cup \Nb.
}
Therefore, if function $f$ is a polynomial of $x$ and $y,$ we have 
\eqlnostar{eq:poly f}{
\Pct_0(t,t_1) f(x,y) &= \int_{\Rb^2} \dd x_1 \dd y_1 f(x_1,y_1) \Gamt_0(t,x,y;t_1,x_1,y_1) \nonumber\\
&= f\bigl(\Xct(t,t_1),\Yct(t,t_1) \bigr)\int_{\Rb^2} \dd x_1 \dd y_1 \Gamt_0(t,x,y;t_1,x_1,y_1)\nonumber\\
&= f\bigl(\Xct(t,t_1),\Yct(t,t_1) \bigr).
}
It is assumed above and throughout the following computations that if an operator is followed by nothing, it acts on the constant 1. Next, we also introduce the following operator
\eqlnostar{eq:op G}{
\Gct_n(t,t_1) := \Act_n(\Xct(t,t_1),\Yct(t,t_1)),\quad t_1 \geq t, n \geq 1,
}
where the notation $\Act_n(\Xct(t,t_1),\Yct(t,t_1))$ indicates that the $(x,y)$-dependence in coefficients $\Act_n \equiv \Act_n(x,y),$ has been replaced with $\bigl( \Xct(t,t_1), \Yct(t,t_1) \bigr).$ For example, in $\Gct_1(t,t_1) = \Act_1(\Xct(t,t_1),\Yct(t,t_1)),$ the term 
\eqstar{
\bigl(\tfrac{1}{2} \sigma^2 \bigr)_1 ( \partial^2_x -\partial_x) = \Bigl(\bigl(\tfrac{1}{2} \sigma^2 \bigr)_{1,0}(x - \bar{x}) + \bigl(\tfrac{1}{2} \sigma^2 \bigr)_{0,1}(y- \bar{y}) \Bigr) ( \partial^2_x -\partial_x) 
}
becomes 
\eqstar{
\Bigl(\bigl(\tfrac{1}{2} \sigma^2 \bigr)_{1,0}(\Xct(t,t_1) - \bar{x}) + \bigl(\tfrac{1}{2} \sigma^2 \bigr)_{0,1}(\Yct(t,t_1)- \bar{y}) \Bigr) ( \partial^2_x -\partial_x).
}
The following result is now in order:
\begin{lemma}[Lemma 3.7 \cite{lorig-4}]
\label{lemma:lorig}
Let the operators $\Act_n, \Pct_0(t,t_1), \Xct(t,t_1), \Yct(t,t_1)$ and $\Gct_n(t,t_1)$ be given by \eqref{eq:n order A}, \eqref{eq:semigroup P0}, \eqref{eq:op X}, \eqref{eq:op Y} and \eqref{eq:op G}, respectively.  Then, we have the following commutation-like relation
\eqstar{
\Pct_0(t,t_1) \Act_n = \Gct_n(t,t_1) \Pct_0(t,t_1),
}
where the operators act on measurable functions $f \in C^{n+2}(\Rb^2)$ whose partial derivatives of all orders less than or equal to $n+2$ are measurable and at most exponentially growing. 
\end{lemma}
To obtain semi-explicit approximation formulas, we fix the European option payoff function as a call option, that is, $\varphi(x) = (\ee^x - \ee^k)^+$ where $k$ is the log strike. Using this choice of $\varphi,$ we obtain the second order approximation of $\psi$ in \eqref{eq:approx psi} using the following result:
\begin{proposition}
\label{prop:approx psi}
Let $\psi_0$ be the unique classical solution of \eqref{eq:solve psi0} and $\psi_1$ and $\psi_2$ be the unique classical solutions of \eqref{eq:solve psim} with source terms $H_1$ \eqref{eq:source H1} and $H_2$ \eqref{eq:source H2}, respectively. Assume that the coefficients $(\tfrac{1}{2}\lambda^2), (\tfrac{1}{2}\sigma^2), (c - \rho \beta \lambda), (\tfrac{1}{2}\beta^2)$ and $( \rho \sigma \beta)$ belong to the class of $C^2(\Rb^2)$ functions. Then, omitting arguments $(x,y)$ for clarity, we have
\eqstar{
\psi_0(t) =& - (\tfrac{1}{2}\lam^2)_0 (T-t) - \gamma \nu  p^{BS}(t),\\
p^{BS}(t) :=& \, \ee^x \Phi\bigl(d_+(t)\bigr) -  \ee^k \Phi\bigl(d_-(t)\bigr), \quad d_{\pm}(t)   = \frac{1}{\sigma_0 \sqrt{T-t}} \bigl(x - k \pm \tfrac{1}{2}\sigma_0^2 (T-t) \bigr),\\
\psi_1(t) =&-\gamma \nu   \int^T_t \dd t_1 \Gct_1(t,t_1) p^{BS}(t) - \int^T_t \dd t_1 (\tfrac{1}{2}\lam^2)_1 \bigl( \Xct(t,t_1),\Yct(t,t_1) \bigr),\\
\psi_2(t) =& -\gamma \nu \Bigl( \int^T_t \dd t_1\int^T_{t_1} \dd t_2 \Gct_1(t,t_1)\Gct_1(t_1,t_2) +  \int^T_t  \dd t_1 \Gct_2(t,t_1) \Bigr) p^{BS}(t)\\
& - \int^T_t \dd t_1 (\tfrac{1}{2}\lam^2)_2 (\Xct(t,t_1),\Yct(t,t_1))  -\int^T_t  \dd t_1 \Gct_1(t,t_1) \int^T_{t_1} \dd t_2(\tfrac{1}{2}\lam^2)_1  \bigl( \Xct(t,t_2),\Yct(t,t_2) \bigr)\\
&+ (1-\rho^2) (\tfrac{1}{2}\beta^2)_0 \left[ (\tfrac{1}{2}\lam^2)^2_{0,1} \tfrac{1}{3}(T-t)^3 + 2 \gamma \nu \int^T_t \dd t_1(\tfrac{1}{2}\lam^2)_{0,1}  (T-t_1) \partial_y \Bigl( \int^T_{t_1} \dd t_2 \Gct_1(t,t_2) \Bigr)  p^{BS}(t) \right.\\
&\left. + \gamma^2 \nu^2 \frac{(\tfrac{1}{2}\sigma^2)^2_{0,1}}{2 \pi \sigma^2_0} \int^T_t \dd t_1 \frac{(T-t_1)^{3/2}}{\sqrt{T-t+t_1 -t}} \exp \left(2 k -  \frac{\bigl( (k-x) + \tfrac{1}{2}\sigma^2_0(T-t) \bigr)^2 }{\sigma^2_0 (T-t+t_1-t)}\right) \right].
}
\end{proposition}
\begin{proof}
See Appendix \ref{sec:proofs}.
\end{proof}
To study the effect of the risk-aversion parameter $\gam$ on the implied Sharpe ratio $\Lambda,$ we derive the asymptotic approximation formula only up to the second order. The second order approximation is sufficient to perform a comparative-static study with respect to the risk-aversion parameter. 
\subsection{Asymptotic approximation formulas for $p$}
\label{sec:asymptotics p}
In an incomplete market, the investors may not agree on their choice of the pricing measure. Thus, price $p$ of the European option with payoff $\varphi$ will depend on the chosen model. In our study, as we fix the model for our reference investor in \eqref{eq:main sde1}-\eqref{eq:main sde2}, it makes sense to obtain $p$ under the appropriate pricing measure based on the model instead of using a market-given European option price. 

We suppose that the pricing measure $\Pbh$ is related to the physical measure $\Pb$ through the following Radon-Nikodym derivative: 

\eqlnostar{eq:likelihood ratio pricing}{
\frac{ \dd \Pbh }{ \dd \Pb } 	&=	\exp \left( -\frac{1}{2} \int_0^T \Bigl(\frac{\mu^2(X_t,Y_t)}{\sig^2(X_t,Y_t)} + \Omega^2(X_t,Y_t) \Bigr)\dd t  - \int_0^T \frac{\mu(X_t,Y_t)}{\sig (X_t,Y_t)} \dd B_t^X - \int_0^T \Omega(X_t,Y_t) \dd B_t^Y\right).
}
Defining $(\Pbh,\Fb)$-Brownian motions $\Bh^X$ and $\Bh^Y$ by
\eqstar{
\Bh^X :=	B_t^X + \int_0^t \frac{\mu(X_s,Y_s)}{\sig(X_s,Y_s)} \dd s, \quad \Bh^Y := B_t^Y + \int_0^t \Omega(X_s,Y_s) \dd s,
}
we see that the dynamics of $(X,Y)$ can be written under $\Pbh$ as
\eqstar{
\dd X_t &=	- \tfrac{1}{2} \sig^2(X_t,Y_t) \dd t + \sig(X_t,Y_t)  \dd \Bh_t^X , \\
\dd Y_t &=	b(X_t,Y_t) \dd t + \beta(X_t,Y_t) \Big( \rho  \dd \Bh_t^X + \sqrt{1-\rho^2}  \dd \Bh_t^Y \Big) ,
}
where the function $b$ is given as
\eqstar{
b(x,y) &:=	c(x,y) -  \rho \beta(x,y) \lambda (x,y) - \sqrt{1- \rho^2} \beta(x,y) \Omega(x,y).
}
The price of a European-style option with payoff function $\varphi$ with $(X,Y)$ having risk-neutral dynamics under $\Pbh$ is given as 
\eqstar{
p(t,x,y) = \Ebh_{t,x,y} \varphi(X_T),
}where $\Ebh$ denotes the expectation operator under $\Pbh$ defined in \eqref{eq:likelihood ratio pricing}. The function $p$ satisfies the linear pricing PDE 
\eqstar{
0  = (\partial_t  + \Ach) p,& & p(T,x,y) = \varphi(x),
}
where $\Ach := \Act - \sqrt{1-\rho^2} \beta \, \Omega\, \partial_y.$ To obtain asymptotic approximation formulas, once again we seek a solution $p^\eps = \sum^\infty_{k=0} \eps^k p_k$ to the following family of PDEs
\eqstar{
0  = (\partial_t  + \Ach^\eps) p^\eps,& & p^\eps(T,x,y) = \varphi(x),
}
where $\Ach^\eps$ is the $\eps$-counterpart of $\Ach$ obtained by replacing the coefficients in it with their $\eps$-counterparts. Upon collecting the terms of like order of $\eps,$ we find that the individual terms satisfy
\eqlnostar{eq:solve p0}{
&\Oc(1):& &0 = (\partial_t  + \Ach_0)p_0, &p_0(T,x) &= \varphi(x),\\
\label{eq:solve pm}
&\Oc(\eps^m):& &0 = (\partial_t  + \Ach_0) p_m + P_m, & p_m(T,x,y) &=0, m \geq 1,
}
where the $m$th-order source term $P_m$ is given by 
\eqlnostar{eq:source p}{
\Oc(\eps^m):& & P_m = \sum^m_{k=1} \Ach_k p_{m-k},
}
and $\Ach_n$ is given by
\eqstar{
\Ach_n := (\tfrac{1}{2}\sig^2)_n (\partial^2_x - \partial_x) + (c - \rho \beta \lam - \sqrt{1-\rho^2} \beta \Omega )_n \partial_y + (\tfrac{1}{2} \beta^2)_n \partial^2_y + (\rho \sig \beta)_n \partial_x \partial_y, \quad n \in \{0\} \cup \Nb.
}
Like $\Act_0$ in Section \ref{sec:asymptotics psi}, $\Ach_0$ gives rise to a semigroup 
$\Pch_0(t,t_1)$ defined by 
\eqstar{
\Pch_0(t,t_1) \eta(x,y) := \int_{\Rb^2} \dd x \dd y \Gamh_0 (t,x,y;t_1,x_1,y_1) \eta(x_1,y_1),
} 
where $\eta: \Rb^2 \to \Rb^2$ is a generic test function and $\Gamh_0$ is the \textit{fundamental solution} to the linear operator $(\partial_t + \Ach_0)$ given as 
\eqstar{
\Gamh_0(t,x,y;t_1,x_1,y_1) = \frac{1}{\sqrt{(2\pi)^2 |\widehat{\mathbf{C}}|}} \exp \left(-\frac{1}{2}\widehat{\mathbf{m}}^\top  \widehat{\mathbf{C}}^{-1} \widehat{\mathbf{m}}\right).
}
The covariance matrix $\widehat{\mathbf{C}}$ and vector $ \widehat{\mathbf{m}}$ are given as 
\eqstar{
\widehat{\mathbf{C}} = (t_1 - t) \begin{pmatrix}
(\sigma^2)_0 & (\rho \sigma \beta)_0\\
(\rho \sigma \beta)_0 &(\beta^2)_0 \\
\end{pmatrix}, \quad \widehat{\mathbf{m}} = (t_1 - t) \begin{pmatrix}
x_1 - x - (t_1 - t)(-\tfrac{1}{2}\sigma^2)_0 \\
y_1 - y - (t_1 - t)(c - \rho \beta \lambda - \sqrt{1-\rho^2}\beta\Omega)_0 \\
\end{pmatrix}.
}
Analogous to $\Xct, \Yct$ and $\Gct$ in Section \ref{sec:asymptotics psi}, we define $\Xch, \Ych$ and $\Gch$ as follows:
\eqstar{
&\Xch(t,t_1) := x + (t_1-t) \left(-(\tfrac{1}{2} \sigma^2)_0 + 2 (\tfrac{1}{2} \sigma^2)_0 \partial_x + (\rho \sigma \beta)_0 \partial_y \right),&  &t_1 \geq t,\\
&\Ych(t,t_1) := y + (t_1-t) \left( (c - \rho \beta \lambda)_0 +  2 (\tfrac{1}{2} \beta^2)_0 \partial_y + (\rho \sigma \beta)_0 \partial_x \right),& &t_1 \geq t,\\
&\Gch_n(t,t_1) := \Ach_n(\Xch(t,t_1),\Ych(t,t_1)),& &t_1 \geq t, n \geq 1.
}
The operators $\Pch_0, \Xch, \Ych$ and $\Gch$ satisfy the same properties as $\Pct_0,\Xct, \Yct$ and $\Gct,$ respectively, in Section \ref{sec:asymptotics psi}. We obtain the following result for approximation of the European option price $p:$
\begin{proposition}
\label{prop:approx p}
Let $p_0$ be the unique classical solution of \eqref{eq:solve p0}, and, $p_1$ and $p_2$ be the unique classical solutions of \eqref{eq:solve pm} with source terms $P_1$ and $P_2$, respectively, obtained from \eqref{eq:source p}. Assume that the coefficients $(\tfrac{1}{2}\lambda^2), (\tfrac{1}{2}\sigma^2), (c - \rho \beta \lambda - \sqrt{1-\rho^2} \beta \Omega), (\tfrac{1}{2}\beta^2)$ and $( \rho \sigma \beta)$ belong to the class of $C^2(\Rb^2)$ functions. Then, omitting arguments $(x,y)$ for clarity, we have 
\eqstar{
p_0(t) =& p^{BS}(t),\\
p^{BS}(t) :=& \, \ee^x \Phi\bigl(d_+(t)\bigr) -  \ee^k \Phi\bigl(d_-(t)\bigr), \quad d_{\pm}(t)   = \frac{1}{\sigma_0 \sqrt{T-t}} \bigl(x - k \pm \tfrac{1}{2}\sigma_0^2 (T-t) \bigr),\\
p_1(t) =&  \int^T_t \dd t_1 \Gch_1(t,t_1) p^{BS}(t),\\
p_2(t) =& \left( \int^T_t \dd t_1 \Gch_2(t,t_1) +  \int^T_t \dd t_1\int^T_{t_1} \dd t_2 \Gch_1(t,t_1)\Gch_1(t_1,t_2)\right) p^{BS}(t).
}
\end{proposition}
\begin{proof}
See Appendix \ref{sec:proofs}.
\end{proof}
\subsection{Asymptotic approximation formulas for $\Lam$}
\label{sec:asymptotics lambda}
We use the results in Proposition \ref{prop:approx psi} and \ref{prop:approx p} to obtain an asymptotic formula for the implied Sharpe ratio $\Lam.$ 
Let us define $\Lam^\eps$ is the positive solution of the following equation
\begin{align}
\gamma \nu p^\eps + \psi^\eps = -(T-t) \tfrac{1}{2} (\Lam^\eps)^2 , \label{eq:Lambda-eps}
\end{align}
which is obtained by replacing $p$ and $\psi$ in \eqref{eq:formula1} by their $\eps$ counterparts.
Recalling that $p^\eps |_{\eps=1} = p$ and $\psi^\eps |_{\eps=1} = \psi$, it follows that $\Lam^\eps |_{\eps=1} = \Lam$.
Thus, in order to find an asymptotic approximation for $\Lam$, we expand $\Lam^\eps$ in powers of $\eps$ as follows
\eqstar{
\Lam^\eps = \sum^\infty_{i=0} \eps^i \Lam_i .
}
Once we obtain expressions for $(\Lam_i)$, our $n$-th order approximation for $\Lam$ will be obtained by truncating the above series at order $n$ and setting $\eps = 1$.
In order to find explicit expressions for the terms$(\Lam_i)$ we insert the expansions for $p^\eps$, $\psi^\eps$ and $\Lam^\eps$ into \eqref{eq:Lambda-eps} and collect terms of like order in $\eps$.  We obtain 
\eqstar{
&\Oc(1):&  \gam \nu p_0 + \psi_0 &=-(T-t)\tfrac{1}{2}\Lam_0^2,& \\
&\Oc(\eps^k):&  \gam \nu p_k + \psi_k  &= -(T-t)\tfrac{1}{2}\Bigl( 2\Lam_0 \Lam_k + \sum_{i+j = k} \Lam_i \Lam_j \Bigr),& k \geq 1.
}
By solving for $\Lam_k,$ we obtain
\eqstar{
&\Oc(1):&  \Lam_0 &= \sqrt{\frac{\gam \nu p_0 + \psi_0}{-\tfrac{1}{2}(T-t)}},& \\
&\Oc(\eps^k):& \Lam_k  &= -\frac{1}{(T-t)\Lam_0}\Bigl( \gamma \nu p_k + \psi_k + (T-t) \tfrac{1}{2} \sum_{i+j=k} \Lam_i \Lam_j \Bigr).
}
\begin{proposition}
\label{prop:approx imp Sharpe}
Let $p_0$ be the unique classical solution of \eqref{eq:solve p0} and $p_1$ and $p_2$ be the unique classical solutions of \eqref{eq:solve pm} with source terms $P_1$ and $P_2$, respectively, obtained from \eqref{eq:source p}. Furthermore, let $\psi_0$ be the unique classical solution of \eqref{eq:solve psi0} and $\psi_1$ and $\psi_2$ be the unique classical solutions of \eqref{eq:solve psim} with source terms $H_1$ \eqref{eq:source H1} and $H_2$ \eqref{eq:source H2}, respectively. Assume that the coefficients $(\tfrac{1}{2}\lambda^2), (\tfrac{1}{2}\sigma^2), (c - \rho \beta \lambda),  (c - \rho \beta \lambda - \sqrt{1-\rho^2}\beta \Omega),  (\tfrac{1}{2}\beta^2)$ and $( \rho \sigma \beta)$ belong to the class of $C^2(\Rb^2)$ functions. Then, the second order approximation of the implied Sharpe ratio, defined as
\eqstar{
\bar{\Lam}_2(t,x,y,w,\nu) := (\Lam_0 + \Lam_1 + \Lam_2)(t,x,y,w,\nu)
} is given explicitly by
\eqstar{
\Lam_0 &= \lambda_0,\\
\Lam_1 &= -\frac{1}{2\Lambda_0(T-t)} (\gamma \nu p_1 + \psi_1) ,\\
\Lam_2  &= -\frac{1}{2(T-t)\Lam_0}\Bigl( \gamma \nu p_2 + \psi_2 + (T-t)\frac{ \Lambda_1^2}{2} \Bigr).
}
\end{proposition}
The above result provides an approximate relationship between $\Lambda \approx \bar{\Lam}_2$, the risk-aversion parameter $\gamma$ and the instantaneous Sharpe ratio $\lambda.$ 
\begin{remark}
If we set $\Omega \equiv 0$ in \eqref{eq:likelihood ratio pricing}, the pricing measure $\Pbh$ corresponds to the \textit{minimal martingale measure} $\Pbt$ as defined in \cite{follmer1991hedging}
%. Thus, the minimal martingale measure $\Pct$ is defined by the following Radon-Nikodym derivative 
\eqstar{
\frac{ \dd \Pbt }{ \dd \Pb } 	&:=	\exp\left( -\frac{1}{2} \int_0^T \frac{\mu^2(X_t,Y_t)}{\sig^2(X_t,Y_t)} \dd t  - \int_0^T \frac{\mu(X_t,Y_t)}{\sig (X_t,Y_t)} \dd B_t^X \right).
}
Moreover, if the market's chosen pricing measure is the minimal marrtingale measure ($\Pbh = \Pbt$), then the operators $\Ach, \Pch, \Xch, \Ych,  \Gch$ defined in Section \ref{sec:asymptotics p} become identical to the operators $\Act, \Pct, \Xct, \Yct,  \Gct,$ respectively, as defined in Section \ref{sec:asymptotics psi}. With $\Pbt$ as the chosen pricing measure, the first and second order correction terms in the approximation of $\Lambda$ simplify to: 
\eqlnostar{eq:imp sharpe first}{
\Lambda_1 =& \frac{1}{2\Lambda_0(T-t)} \int^T_t \dd t_1 (\tfrac{1}{2} \lam^2)_1 (\Xct(t,t_1), \Yct(t,t_1) ),\\
\label{eq:imp sharpe second}
\Lambda_2 =&  \frac{1}{2(T-t)\Lam_0}\Biggl( \int^T_t \dd t_1 (\tfrac{1}{2}\lam^2)_2 (\Xct(t,t_1),\Yct(t,t_1))  + \int^T_t  \dd t_1 \Gct_1(t,t_1) \int^T_{t_1} \dd t_2(\tfrac{1}{2}\lam^2)_1  \bigl( \Xct(t,t_2),\Yct(t,t_2) \bigr)\nonumber\\
&- (1-\rho^2) (\tfrac{1}{2}\beta^2)_0 \Bigl[ (\tfrac{1}{2}\lam^2)^2_{0,1} \tfrac{1}{3}(T-t)^3 + 2 \gamma \nu \int^T_t \dd t_1(\tfrac{1}{2}\lam^2)_{0,1}  (T-t_1) \partial_y \Bigl( \int^T_{t_1} \dd t_2 \Gct_1(t,t_2) \Bigr)  p^{BS}(t) \nonumber\\
& + \gamma^2 \nu^2 \frac{(\tfrac{1}{2}\sigma^2)^2_{0,1}}{2 \pi \sigma^2_0} \int^T_t \dd t_1 \frac{(T-t_1)^{3/2}}{\sqrt{T-t+t_1 -t}} \exp \Bigl(2 k -  \frac{\bigl( (k-x) + \tfrac{1}{2}\sigma^2_0(T-t) \bigr)^2 }{\sigma^2_0 (T-t+t_1-t)}\Bigr) \Bigr] - (T-t)\frac{ \Lambda_1^2}{2}   \Biggr) ,
}
and the zeroth order term $\Lam_0$ remains unaffected.
\end{remark}
\section{Examples}
\label{sec:examples}
In this section, we consider different local stochastic volatility models which provide different functional forms of $\lambda.$ Using our approximation result in Proposition \ref{prop:approx imp Sharpe}, we discuss its practical implications. For the purpose of simplification of presentation, we will assume that the pricing measure $\Pbh$ corresponds to the minimal martingale measure $\Pbt.$
\subsection{Heston model}
\label{sec:heston}
We first consider the famous Heston's stochastic volatility model, which under the physical measure $\Pb,$ is given as 
\eqstar{
\dd X_t &= \left( \lambda(X_t,Y_t) \sqrt{Y_t} - \tfrac{1}{2}Y_t \right) \dd t + \sqrt{Y_t} \dd B^X_t,\\
\dd Y_t &= \left( \kappa(\theta - Y_t) + \rho \delta \lambda(X_t,Y_t) \sqrt{Y_t} \right) \dd t +\delta \sqrt{Y_t}\left( \rho  \dd B^X_t + \sqrt{1-\rho^2_t} \dd B^Y_t  \right).
}
Comparing the above model with our formulation in \eqref{eq:main sde1}, we have 
\eqstar{
&\mu(x,y) = \lambda(x,y) \sqrt{y}, &\sigma(y) = \sqrt{y},\\
&c(x,y) = \kappa (\theta - y) + \rho \delta \lambda(x,y) \sqrt{y}, & \beta(y) = \delta \sqrt{y}.
}

We suppose the following form of $\lambda(x,y) = -\frac{\sqrt{y}}{2} + \frac{\sqrt{\theta}}{3}.$ This particular choice keeps the model in the affine class (see \cite{duffee2002term}). Other forms of the function $\lambda(x,y)$ can also be considered. For the chosen $\lambda$, we plot the second order approximation $\sum^2_{i=0}\Lambda_i$ of the implied Sharpe ratio in Figure \ref{fig:gamma heston} with respect to the risk-aversion parameter $\gamma$ for $\nu = \pm1, \pm2, \pm3,\pm4.$ We observe that including European options in the investment portfolio increases the implied Sharpe ratio and its impact is greater for an investor with a higher value of $\gamma$. Thus, a risk-averse investor is better off including an European option in his portfolio when considering to invest using the utility maximisation approach. 
\begin{figure}[H]
\centering
\subfigure[]{\includegraphics[scale=0.4]{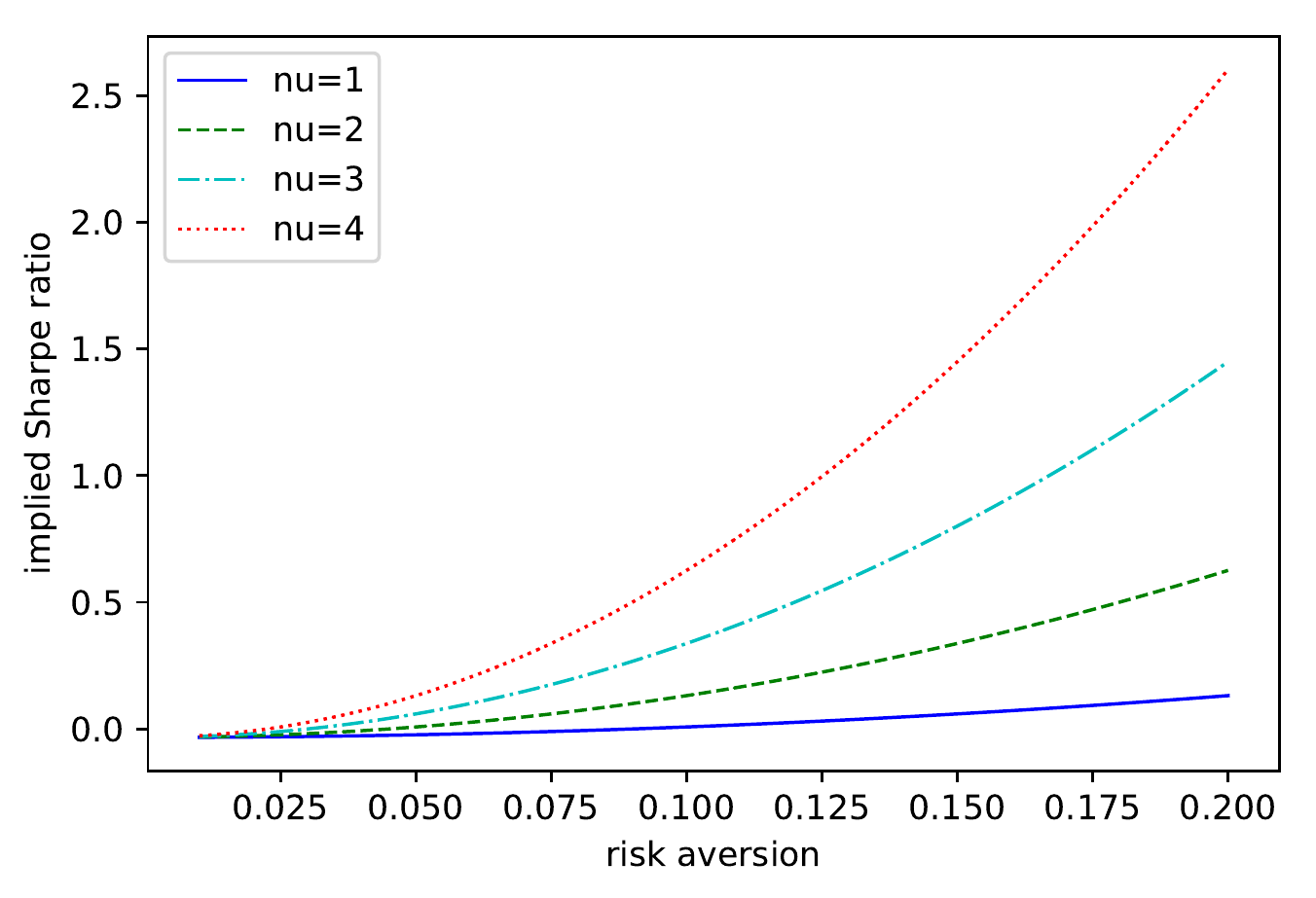}}
%l b r t
~
\subfigure[]{\includegraphics[scale=0.4]{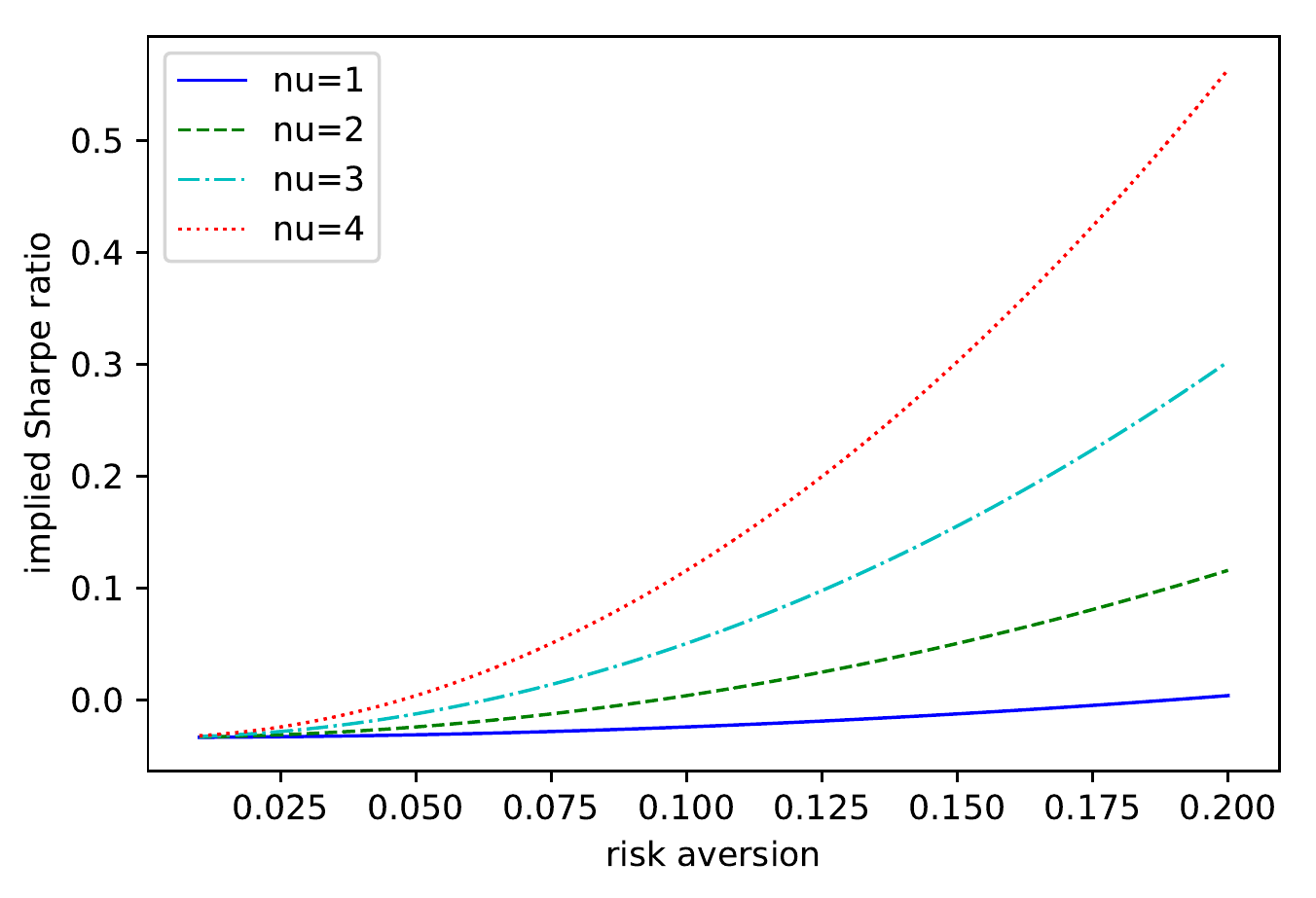}}
%l b r t
~
\subfigure[]{\includegraphics[scale=0.4]{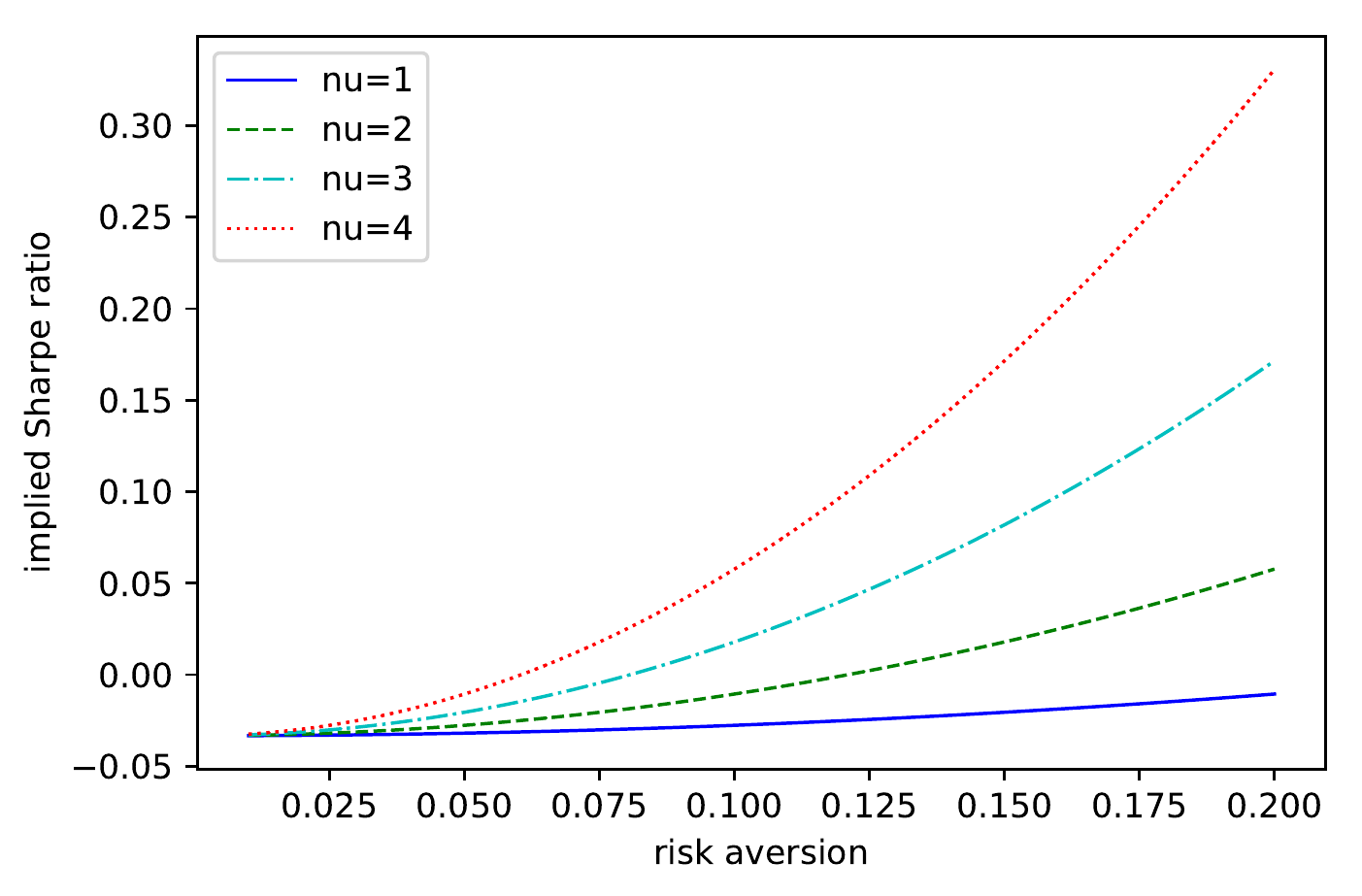}}
%l b r t
~
\caption{{ \small Implied Sharpe ratio for different values of log price (a) $x=\log(100)$ (b) $x=\log(110)$ (c) $x=\log(90).$ The parameter values used are $k=\log(100), t=0,T=6/52, \delta = 0.2, \theta = 0.04, \kappa = 1.15, \rho = -0.4,  \bar{x} = x, \bar{y} = \theta,$ and $y = \bar{y}.$}}
\label{fig:gamma heston}
\end{figure} 
\begin{figure}[H]
\centering
\subfigure[]{\includegraphics[scale=0.4]{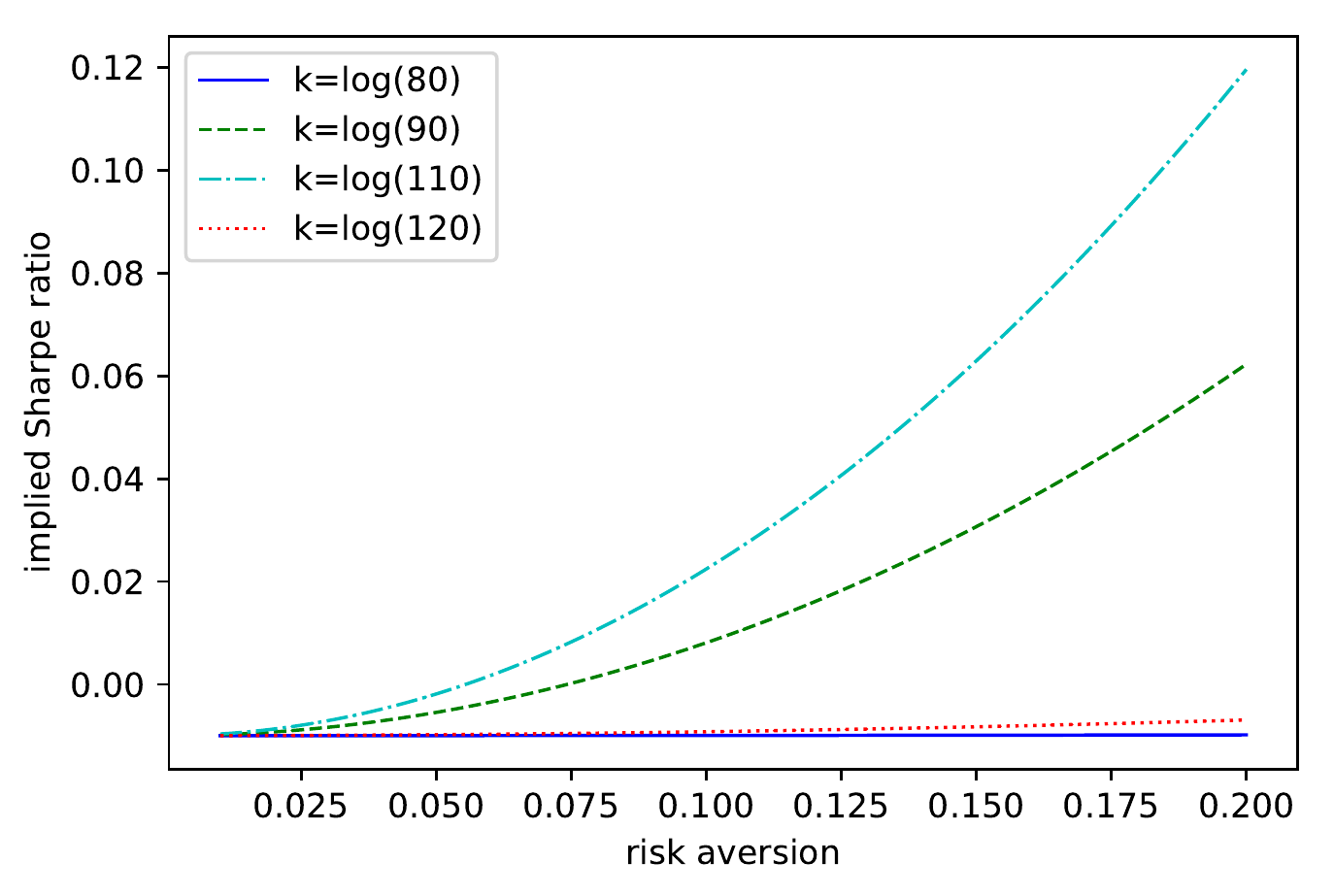}}
%l b r t
~
\subfigure[]{\includegraphics[scale=0.4]{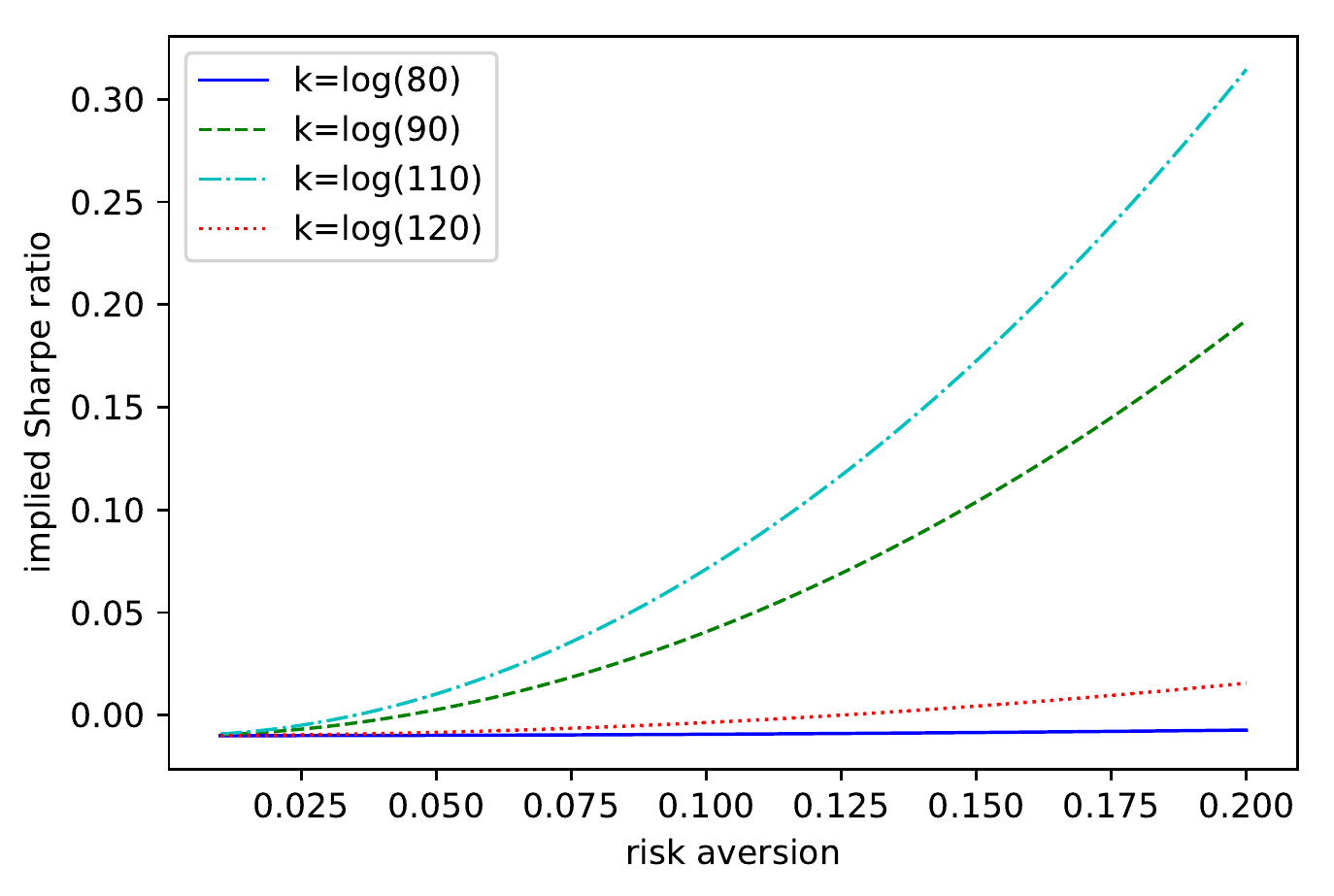}}
%l b r t
~
\subfigure[]{\includegraphics[scale=0.4]{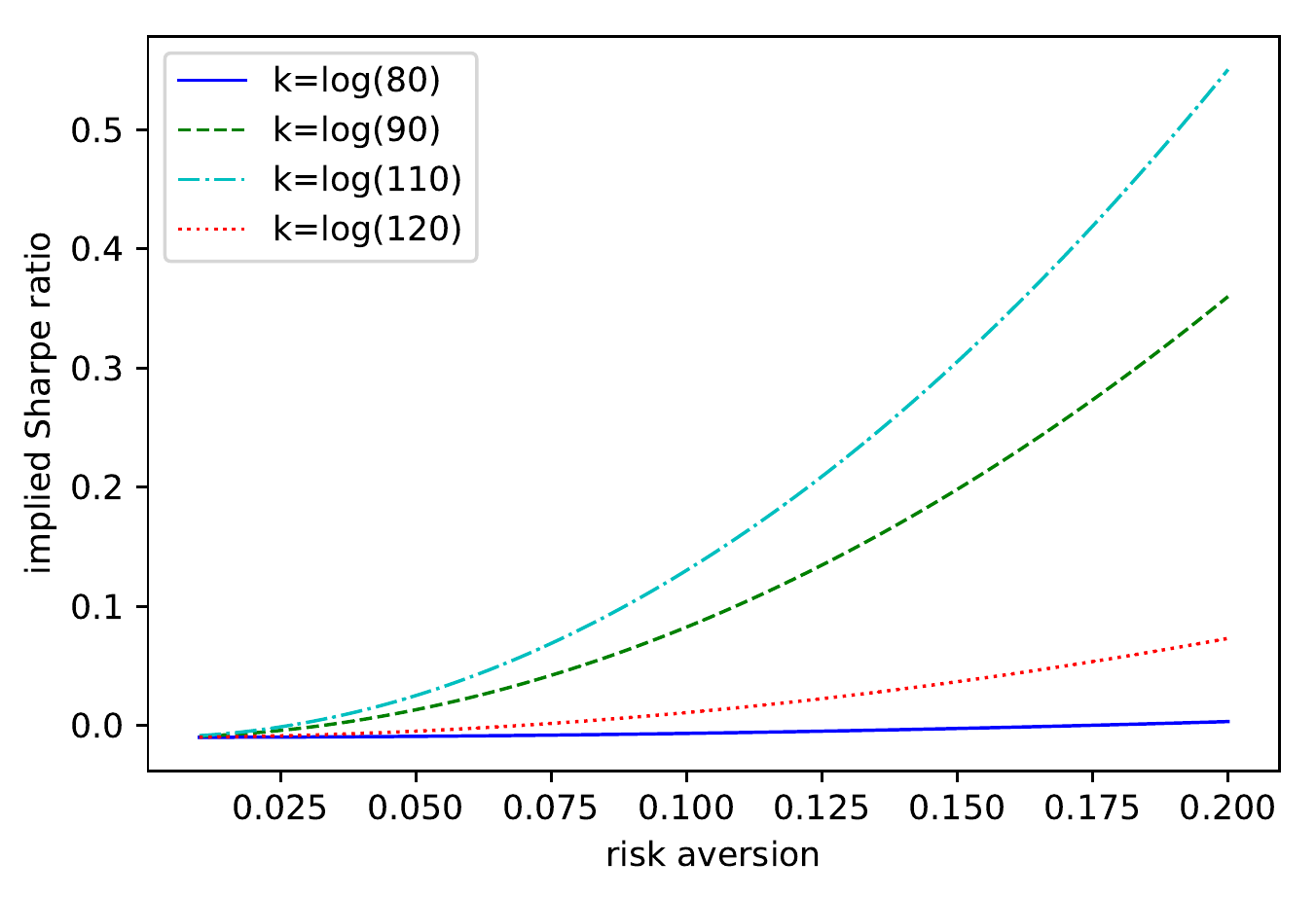}}
%l b r t
~
\caption{{\small Relationship of the implied Sharpe ratio with respect to log-strike (a) $T=6/52$ (b) $T=9/52$ (c) $T=12/52.$ The parameter values used are $t=0, \nu = 1, \delta = 0.2,x=\log(10), \bar{x} = x, \theta = 0.04, \kappa = 1.15, \rho = -0.4, \bar{y} = \theta,$ and $y = \bar{y}.$}}
\label{fig:strike heston}
\end{figure} 
To compare different European options we plot the second order approximation of the implied Sharpe ratio with respect to the risk-aversion parameter $\gamma$ for different values of log-strike $k$ and maturity $T$ in Figure \ref{fig:strike heston} and Figure \ref{fig:maturity heston}, respectively. In Figure \ref{fig:strike heston}, we observe that for a fixed maturity and the chosen parameter values, near-the-money European call options provide better implied Sharpe ratio than the far-from-the-money European options. Moreover, this impact is more pronounced for an investor with higher value of $\gamma.$ Thus, a risk-averse investor should consider to include near-the-money European options in the investment portfolio.
\begin{figure}[htbp]
\centering
\subfigure[]{\includegraphics[scale=0.4]{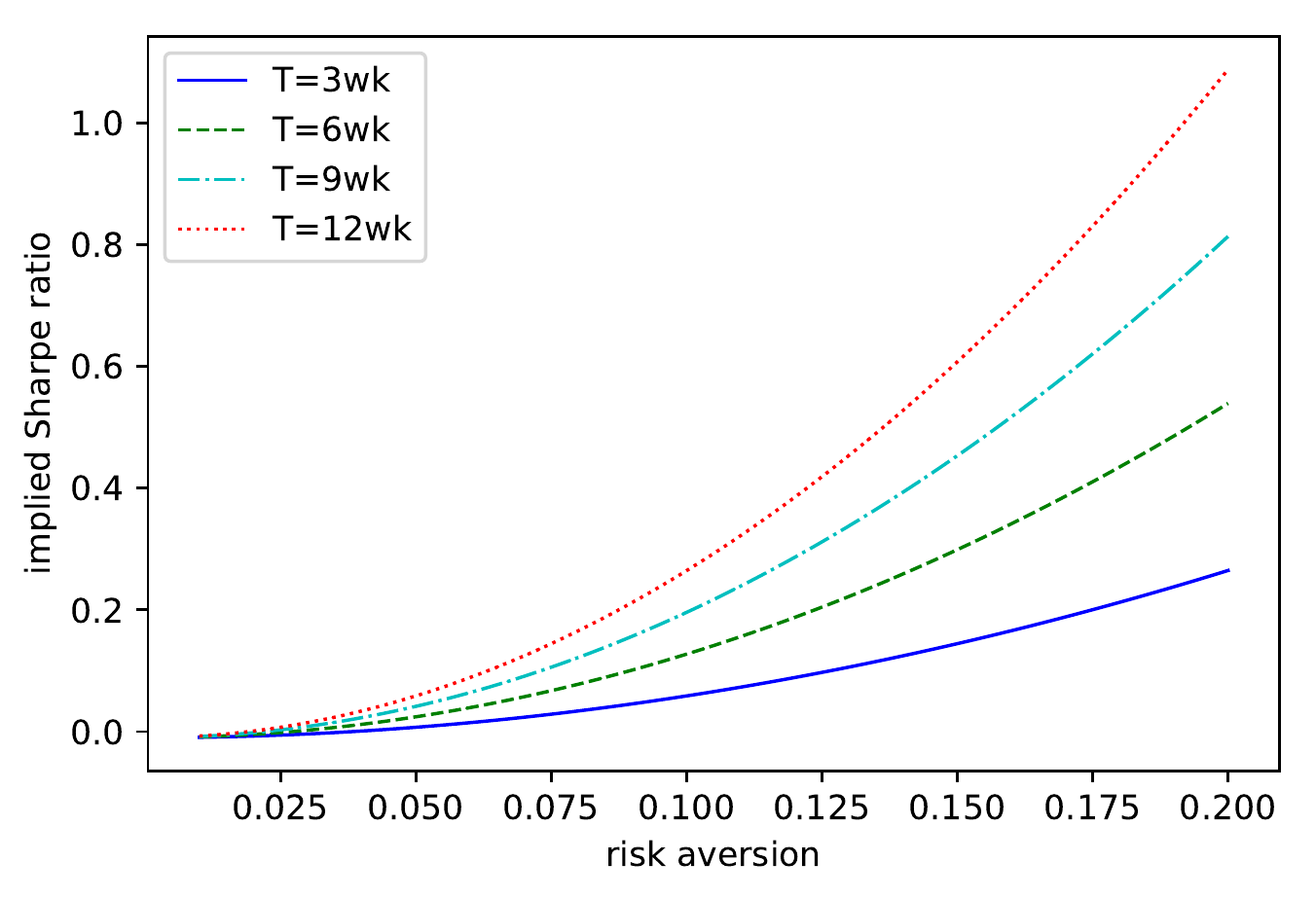}}
%l b r t
~
\subfigure[]{\includegraphics[scale=0.4]{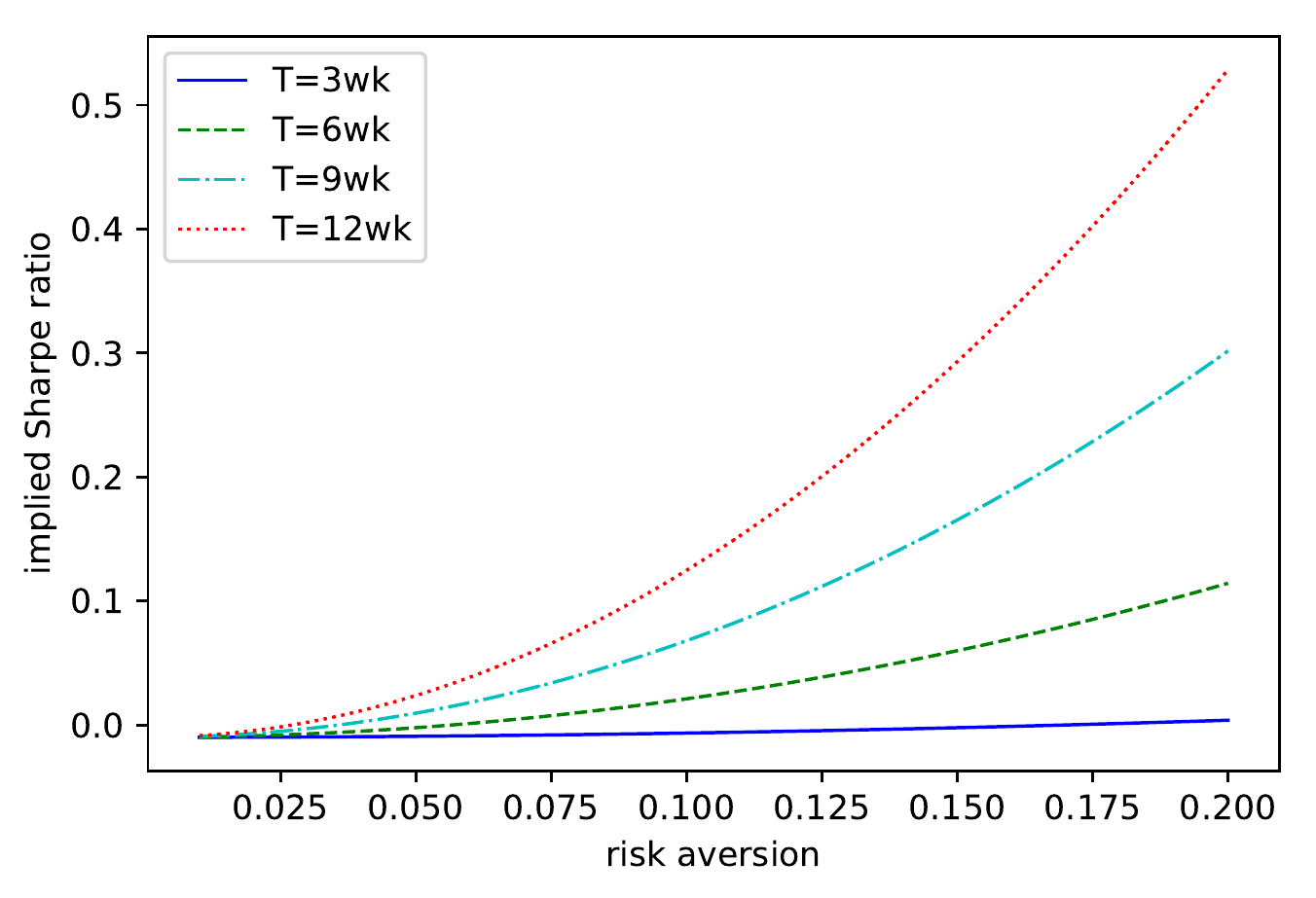}}
%l b r t
~
\subfigure[]{\includegraphics[scale=0.4]{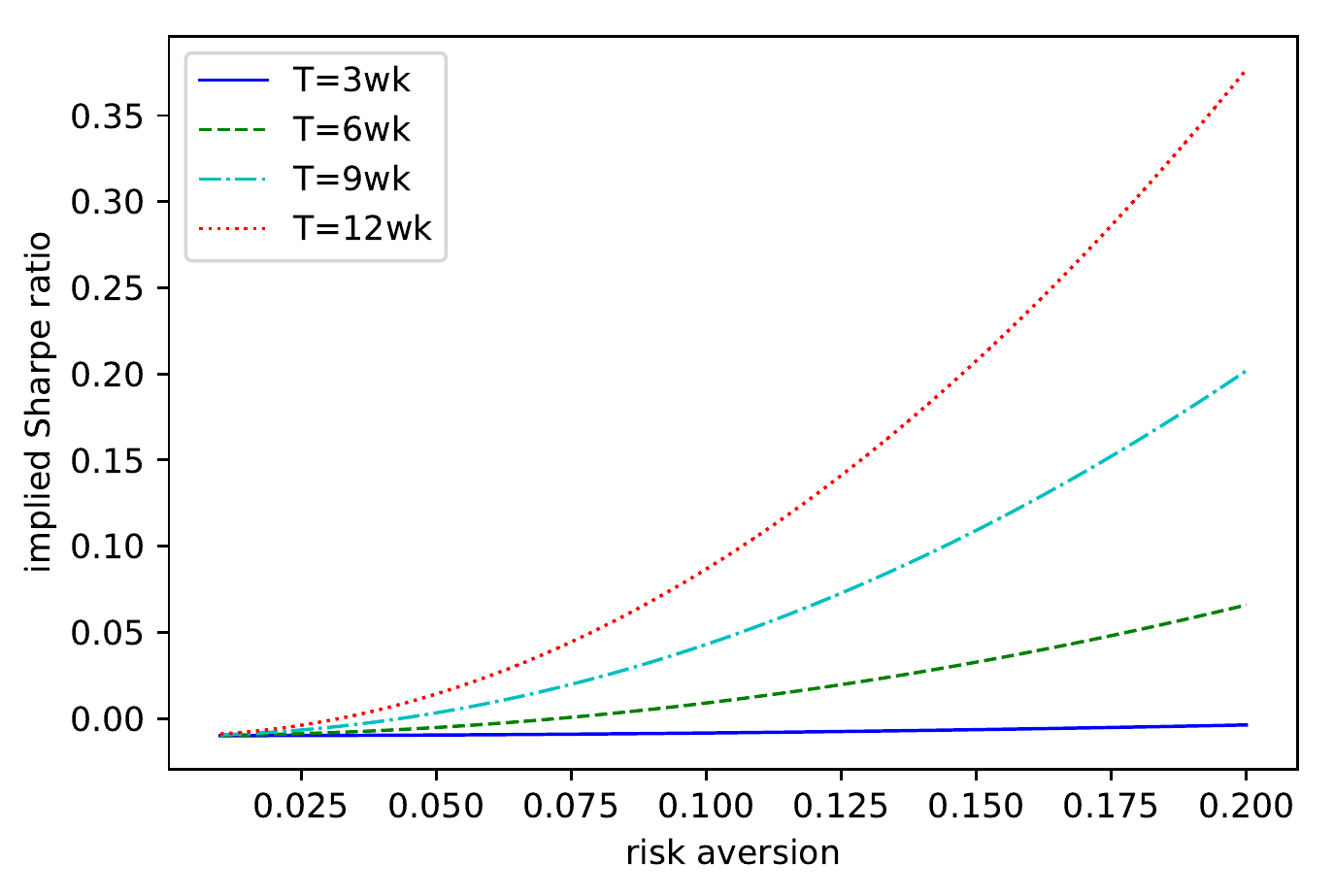}}
%l b r t
~
\caption{{\small Relationship of the implied Sharpe ratio with respect to maturity (a) $k=\log(100),x=\log(100)$ (b) $k=\log(100),x=\log(110)$ (c) $k=\log(100),x=\log(90).$ The parameter values used are $t=0, \delta = 0.2, \nu = 1, \theta = 0.04, \kappa = 1.15, \rho = -0.4,  \bar{x} = x, \bar{y} = \theta,$ and $y = \bar{y}.$}}
\label{fig:maturity heston}
\end{figure} 
In Figure \ref{fig:maturity heston}, we observe that for the chosen parameter values, the investor's implied Sharpe ratio increases with increasing maturity irrespective of the moneyness of the European call option, with an increasing impact for higher values of $\gamma$. Thus, under the Heston model with the chosen form of $\lambda,$ the risk-averse investor should consider to include European options with longer time-to-maturity over options with shorter time-to-maturity. 

\subsection{Reciprocal Heston model}
\label{sec:reciprocal heston}
We consider another stochastic volatility model, which under the physical measure $\Pb$ is given as 
\eqstar{
\dd X_t &= \left( \mu - \tfrac{1}{2}Y_t \right) \dd t + \sqrt{Y_t} \dd B^X_t,\\
\dd Y_t &= \left( a Y_t +\frac{2 (b^2 - a \kappa)}{\mu^2 ( 1- \rho)^2} Y_t^2 \right) \dd t - \left(\frac{2}{1- \rho^2	}\right)^{1/2} \frac{b}{\mu} Y_t^{3/2} \left( \rho  \dd B^X_t + \sqrt{1-\rho^2_t} \dd B^Y_t  \right).
}
The above model is referred as \textit{reciprocal Heston model} as $Y$ is the reciprocal of a CIR process. Comparing the above model with our formulation in \eqref{eq:main sde1}, we get that 
\eqstar{
&\mu(y) =\mu,& &\sigma(y) = \sqrt{y},&\\
&c(y) = a y + \frac{2 (b^2 - a \kappa)}{\mu^2 ( 1- \rho)^2} y^2,& & \beta(y) = - \left(\frac{2}{1- \rho^2	}\right)^{1/2} \frac{b}{\mu} y^{3/2},&
}
where $(a,b,\kappa)$ must satisfy the usual Feller condition: $2 a \kappa  \geq b^2.$ This model choice leads to the functional choice of $\lambda(x,y) = \frac{\mu}{\sqrt{y}}$ which is different from the choice of $\lambda$ in Section \ref{sec:heston}. Once again we plot the second order approximation of the implied Sharpe ratio with respect to the risk-aversion parameter $\gamma$ in Figure \ref{fig:gamma rec heston}. We observe that in the reciprocal Heston model, the implied Sharpe ratio of a risk-averse investor increases by including the European option in the investment portfolio. 
\begin{figure}[H]
\centering
\subfigure[]{\includegraphics[scale=0.4]{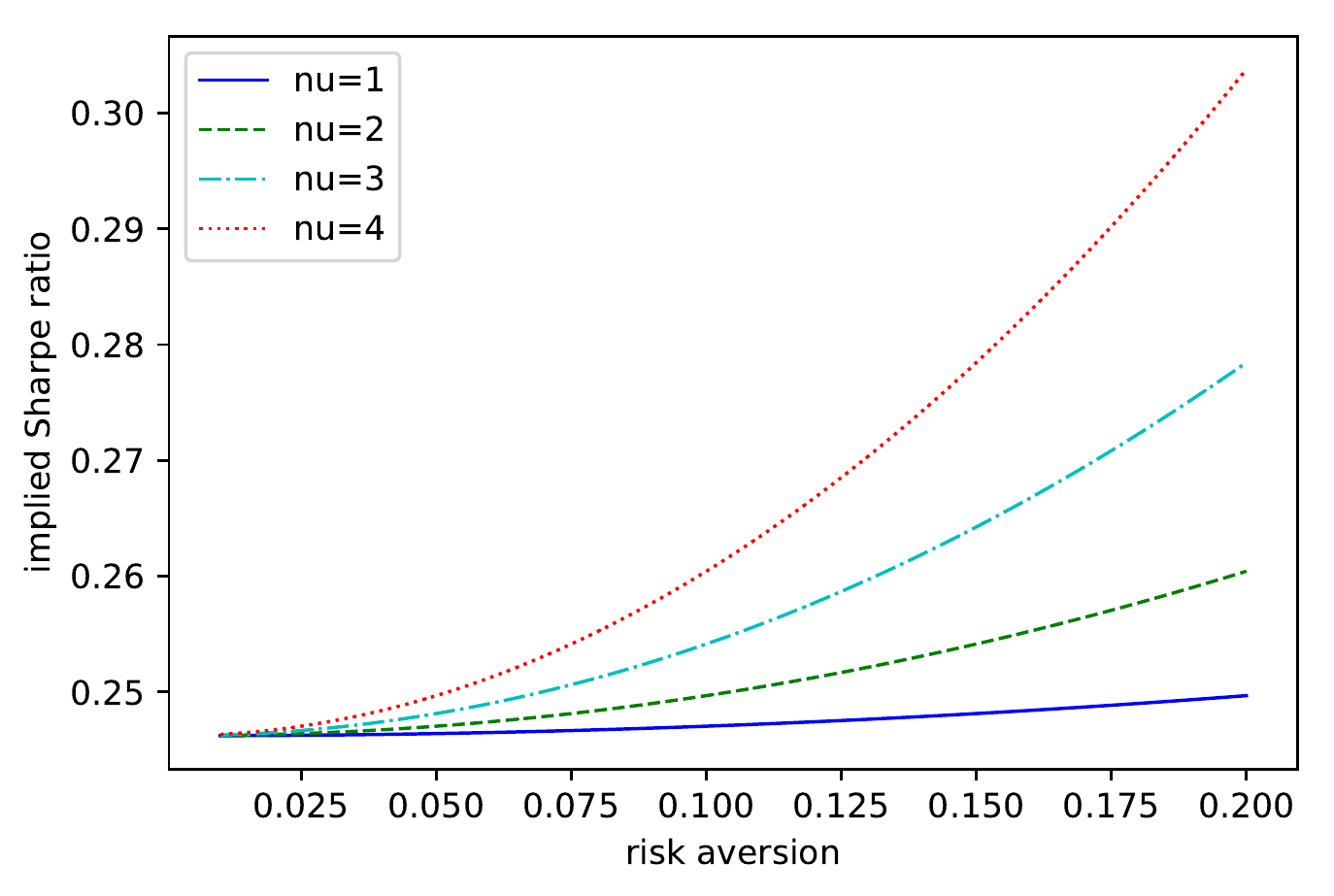}}
%l b r t
~
\subfigure[]{\includegraphics[scale=0.4]{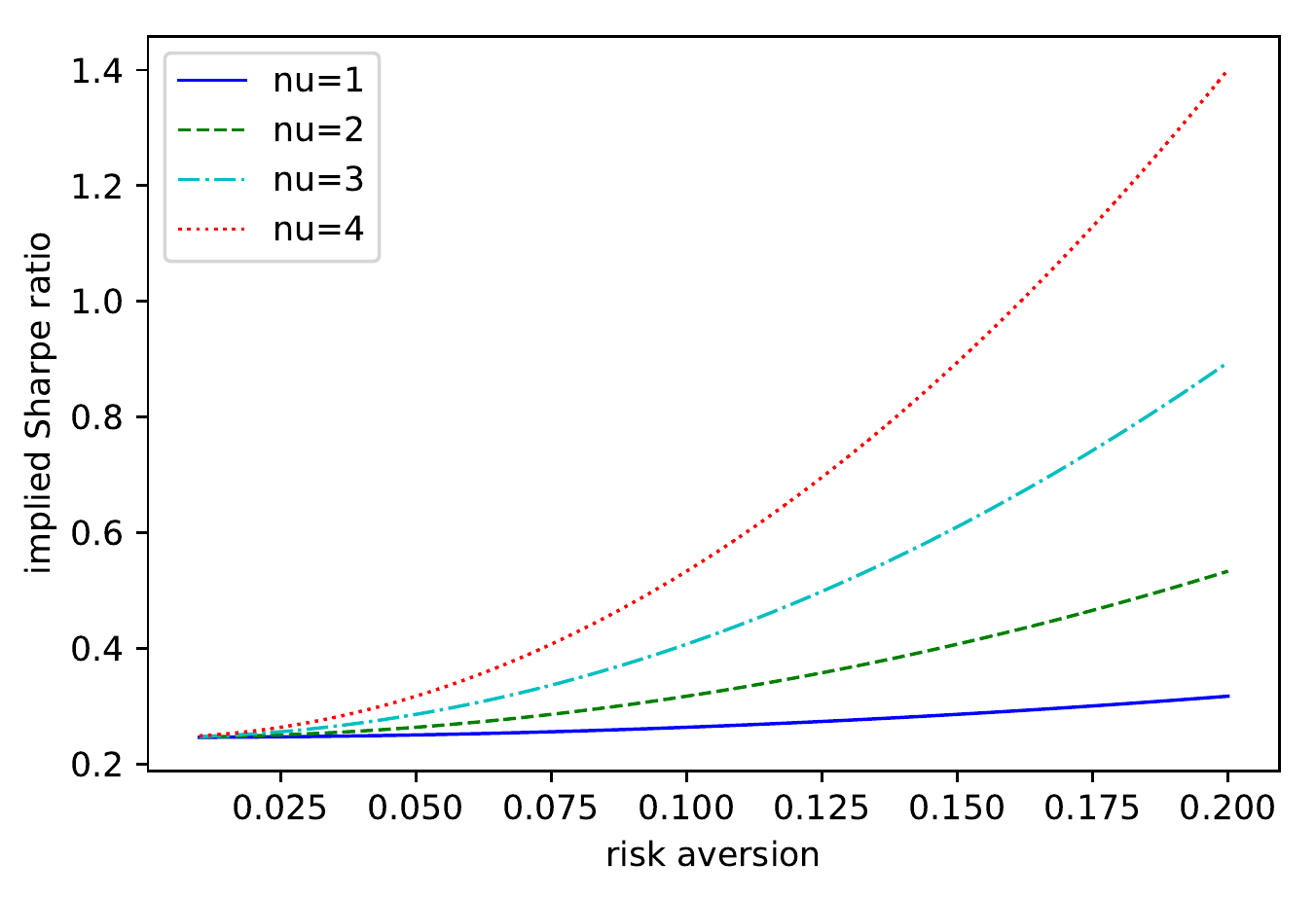}}
%l b r t
~
\subfigure[]{\includegraphics[scale=0.4]{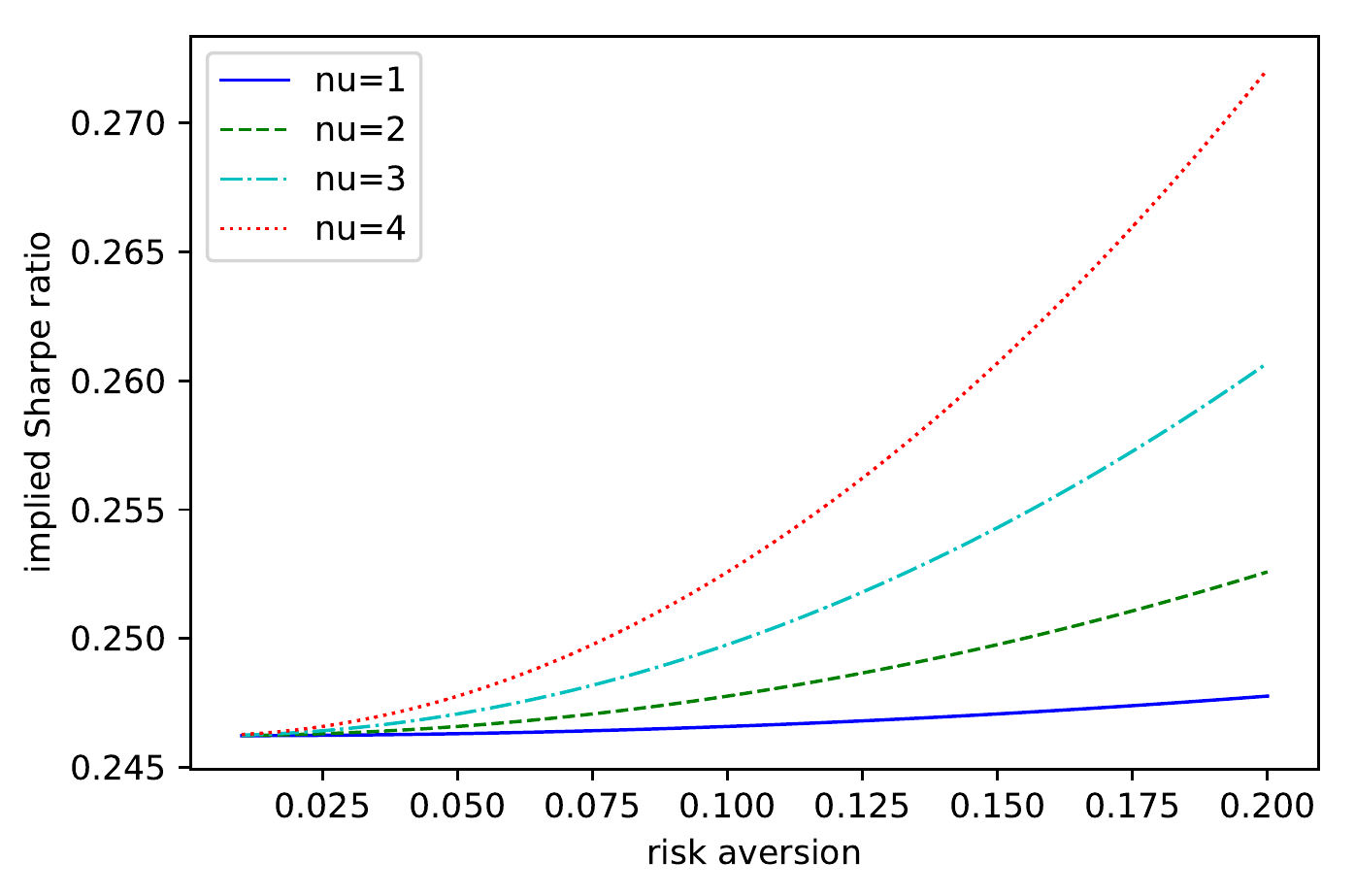}}
%l b r t
~
\caption{{\small Implied Sharpe ratio for different values of log price (a) $k=\log(100),x=\log(110)$ (b) $k=\log(100),x=\log(100)$ (c) $k=\log(100),x=\log(90).$ The parameter values used are $t=0,T=0.25, \mu = 0.05, a = 5.0, b = 0.04, \kappa = 0.01, \rho = 0.2, \bar{y} = 0.04,$ and $y = \bar{y}.$}}
\label{fig:gamma rec heston}
\end{figure}
\begin{figure}[H]
\centering
\subfigure[]{\includegraphics[scale=0.4]{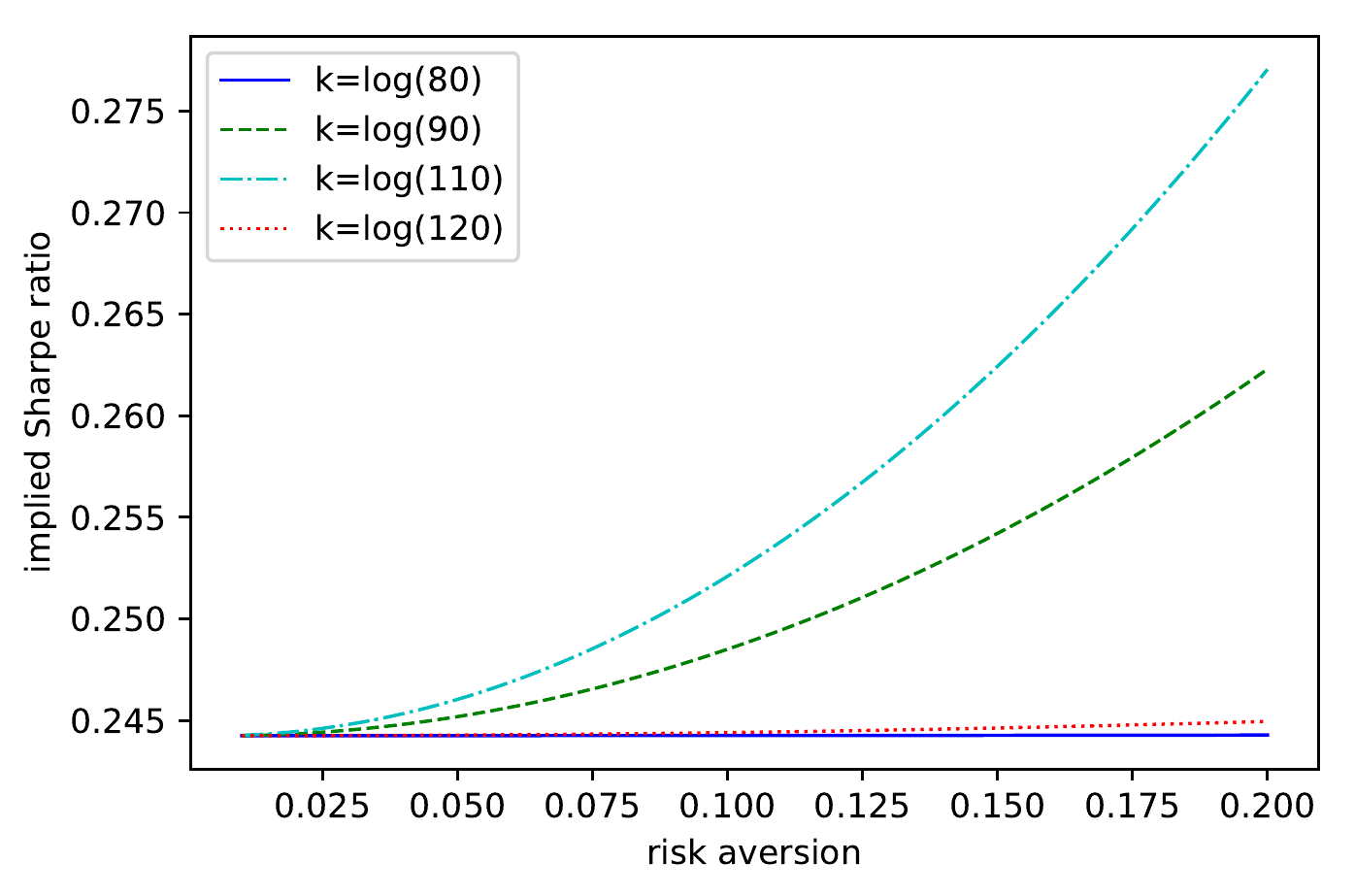}}
%l b r t
~
\subfigure[]{\includegraphics[scale=0.4]{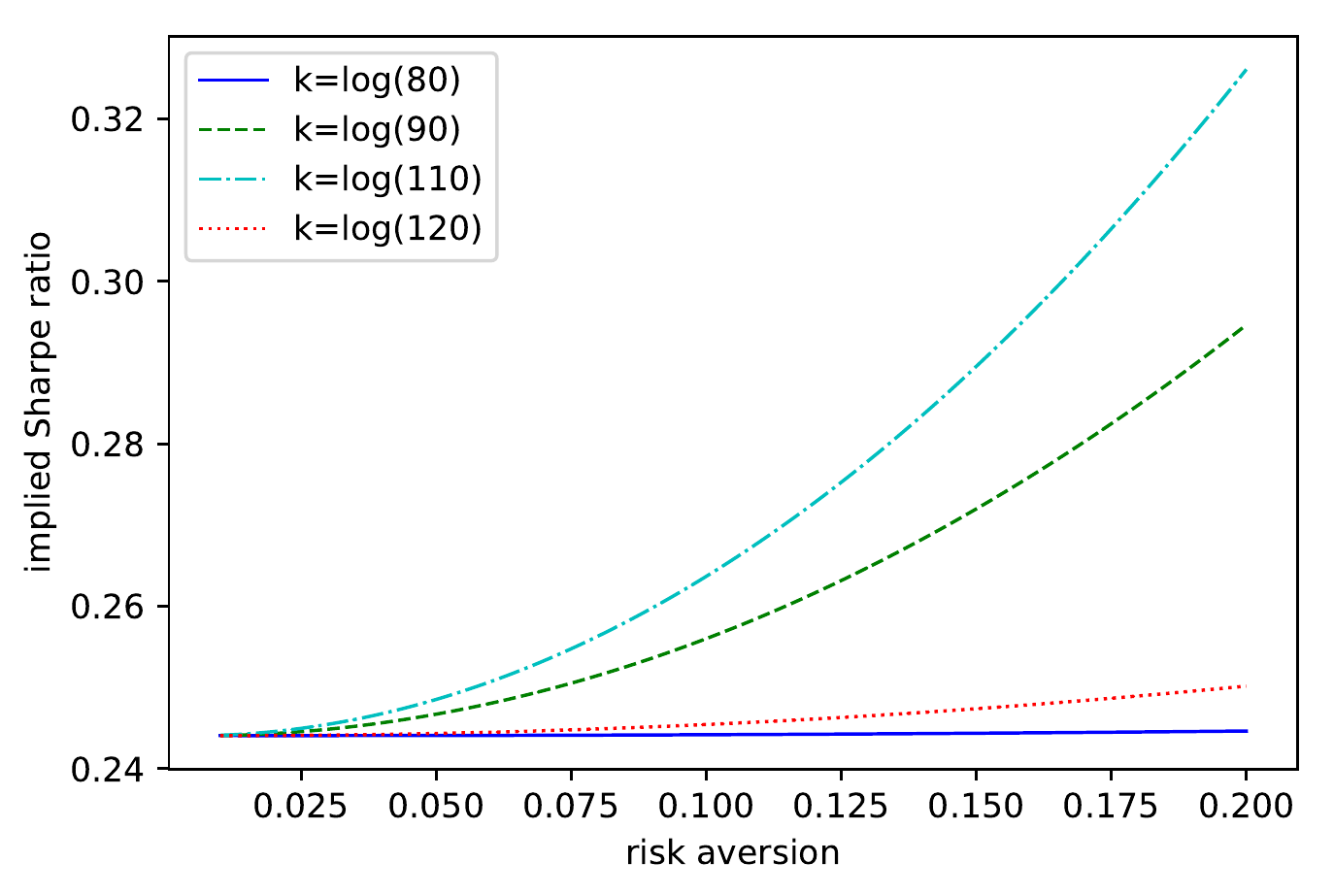}}
%l b r t
~
\subfigure[]{\includegraphics[scale=0.4]{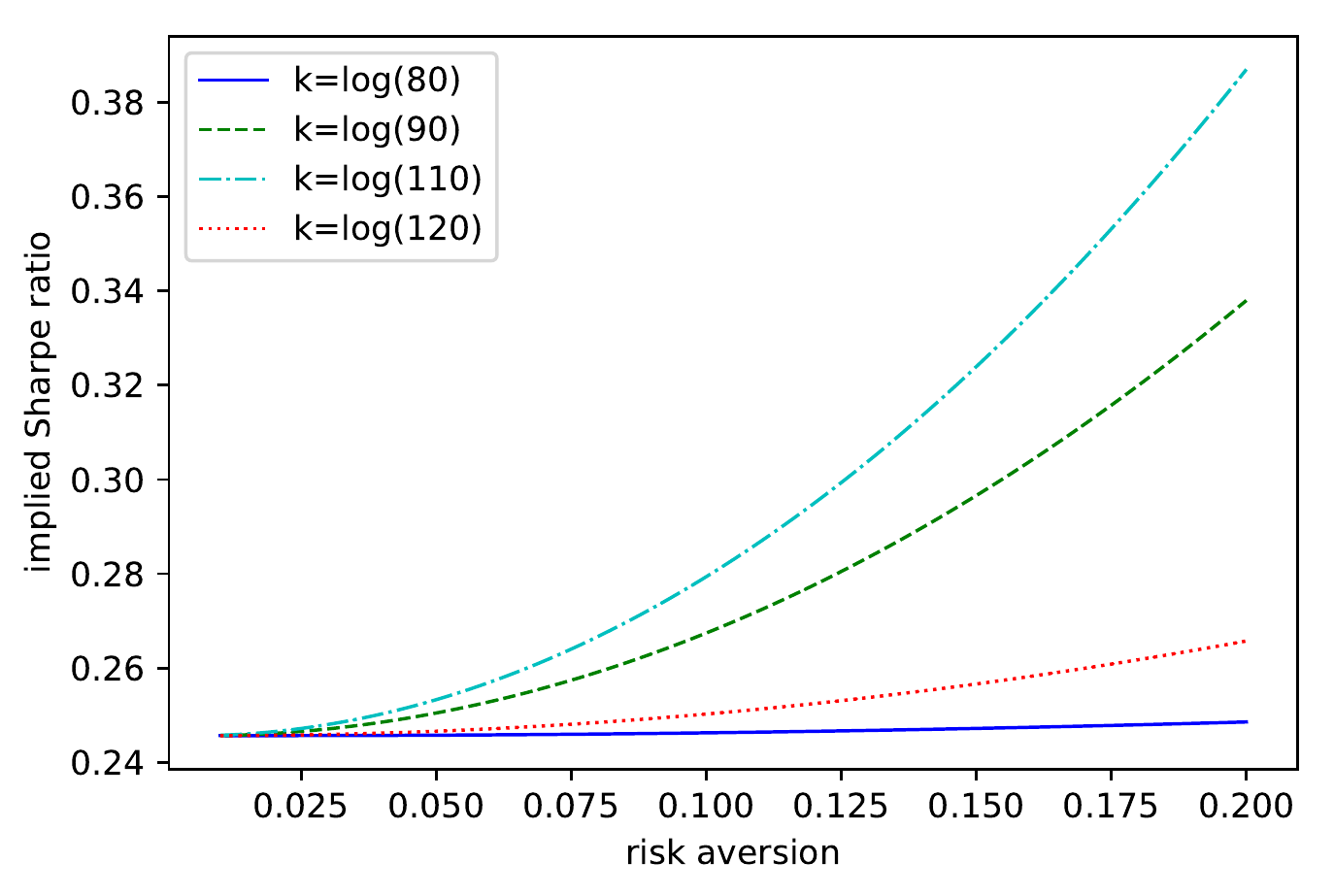}}
%l b r t
~
\caption{{\small Relationship of the implied Sharpe ratio with respect to log-strike (a) $T=6/52$ (b) $T=9/52$ (c) $T=12/52.$ The parameter values used are $t=0, \nu = 1, x=\log(10), \bar{x} = x,\mu = 0.05, a = 5.0, b = 0.04, \kappa = 0.01, \rho = 0.2, \bar{y} = 0.04,$ and $y = \bar{y}.$}}
\label{fig:strike rec heston}
\end{figure} 
We also compare different European options by plotting the second order approximation of the implied Sharpe ratio with respect to the risk-aversion parameter $\gamma$ for different values of log-strike $k$ and maturity $T$ in Figure \ref{fig:strike rec heston} and Figure \ref{fig:maturity rec heston}, respectively. In Figure \ref{fig:strike rec heston}, we observe that for a fixed maturity and the chosen parameter values, far-from-the-money European call options provide lower implied Sharpe ratio compared to the near-the-money options. 
\begin{figure}[htbp]
\centering\subfigure[]{\includegraphics[scale=0.4]{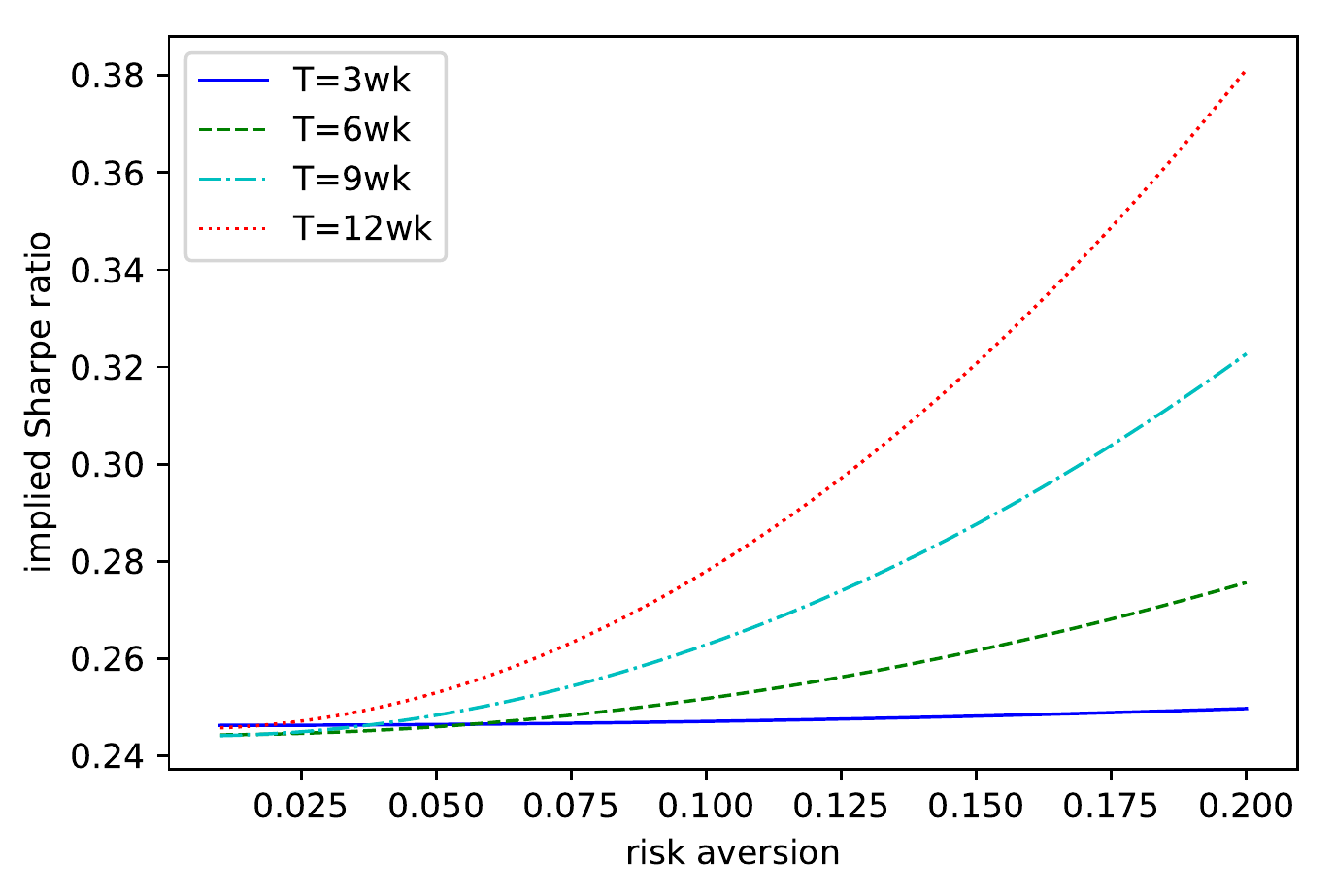}}
%l b r t
~
\subfigure[]{\includegraphics[scale=0.4]{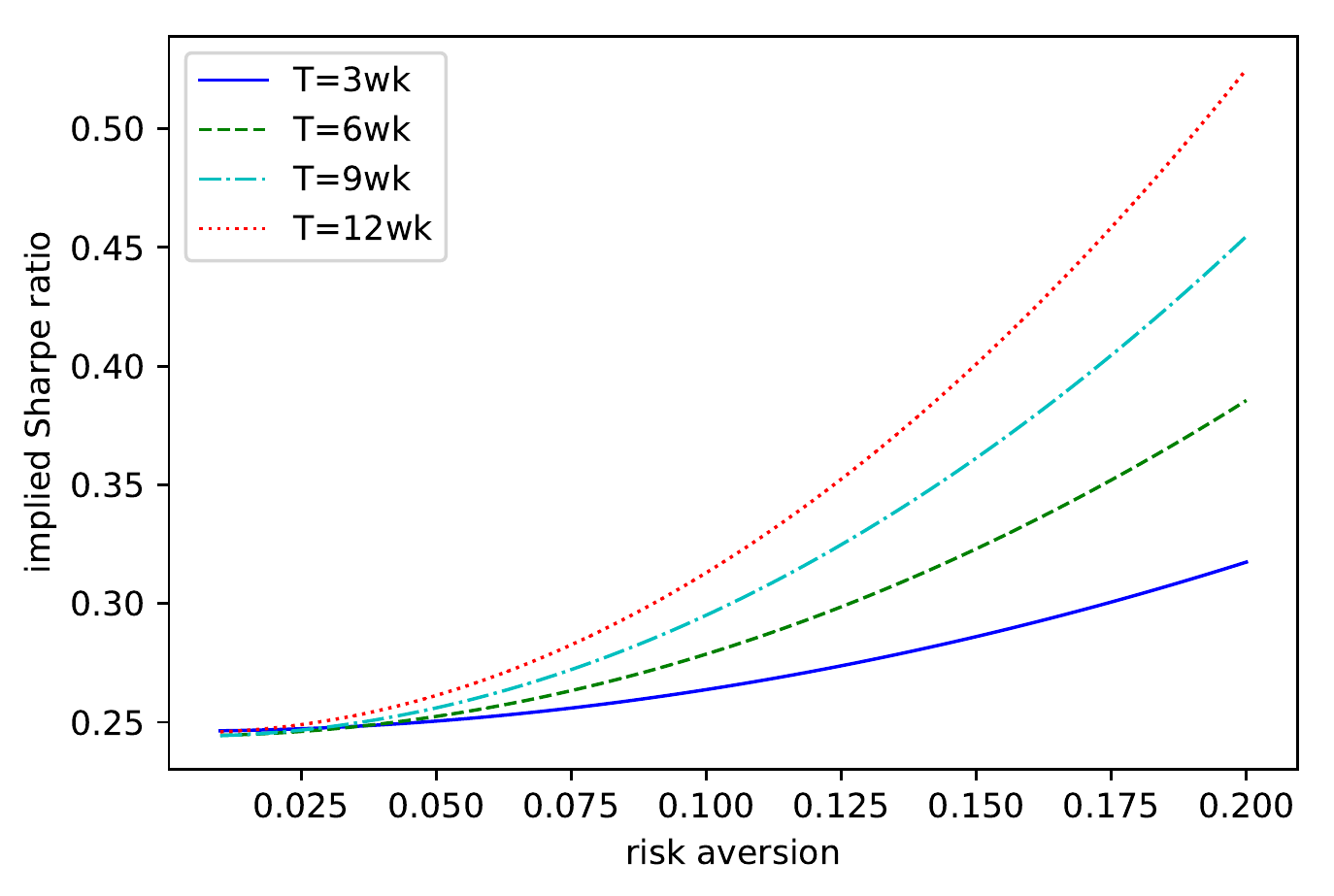}}
%l b r t
~
\subfigure[]{\includegraphics[scale=0.4]{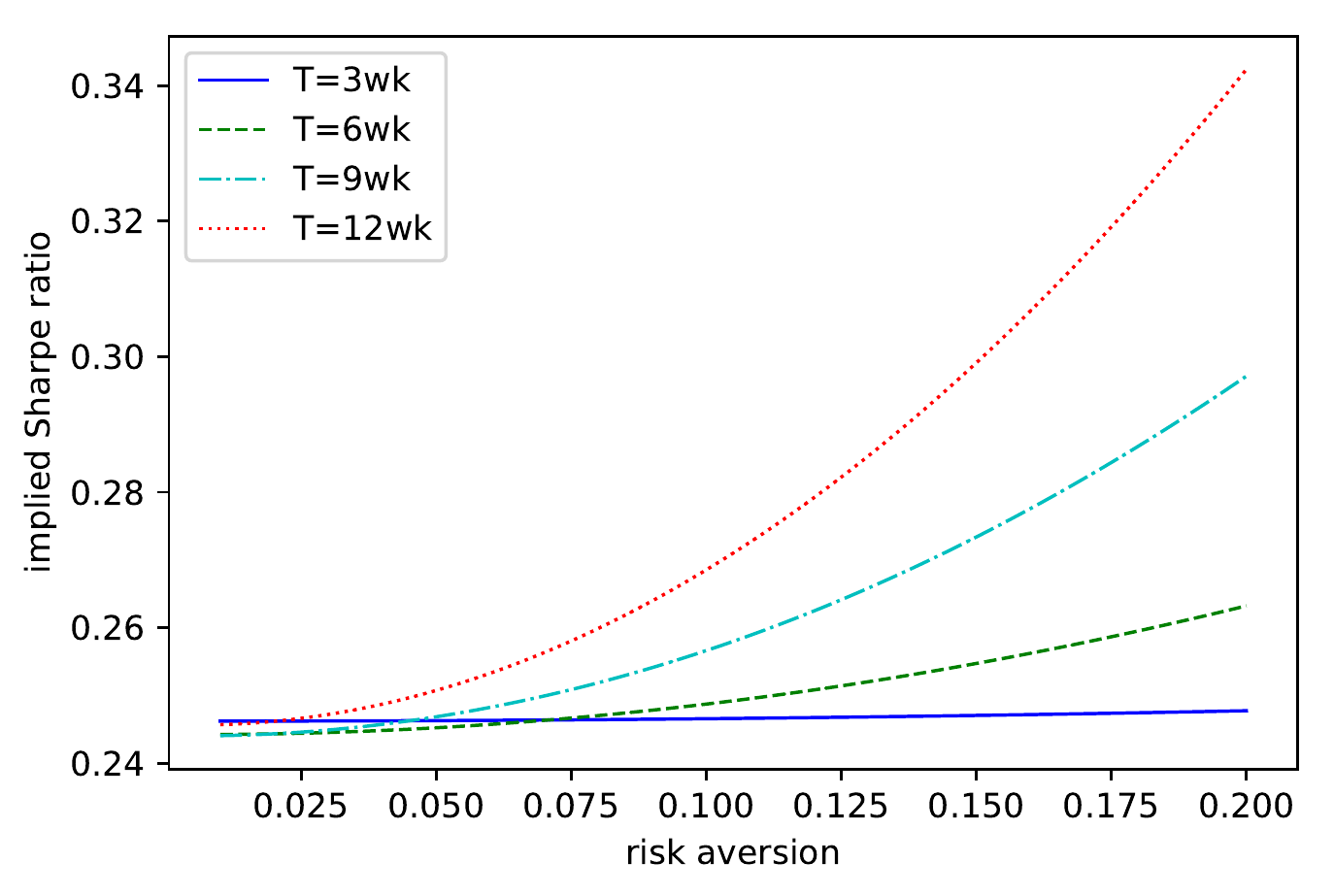}}
%l b r t
~
\caption{{\small Relationship of the implied Sharpe ratio with respect to maturity (a) $k=\log(100),x=\log(110)$ (b) $k=\log(100),x=\log(100)$ (c) $k=\log(100),x=\log(90).$ The parameter values used are  $t=0, \nu = 1, \bar{x} = x,\mu = 0.05, a = 5.0, b = 0.04, \kappa = 0.01, \rho = 0.2, \bar{y} = 0.04,$ and $y = \bar{y}.$}}
\label{fig:maturity rec heston}
\end{figure} 
In Figure \ref{fig:maturity rec heston}, apart from the very small values of the risk-aversion parameter, we observe that the investor's implied Sharpe ratio is higher for longer maturity options. Both the behaviours are similar to the observations under the Heston model in Section \ref{sec:heston}.

\section{Conclusion}
\label{sec:conclusion}
In this work, we introduced a new concept of the implied Sharpe ratio which allows a risk-averse investor to compare the worthiness of different European options for investment. In a general market setting, we prove the existence and uniqueness of the implied Sharpe ratio and derive its asymptotic approximation formulas under general local stochastic volatility models. In two stochastic volatility model settings, we used the implied Sharpe ratio approximation formulas to observe that including a European option in the investment portfolio increases the investor's utility. Moreover, we observed that near-the-money and longer maturity options provide higher implied Sharpe ratio than far-from-the-money and shorter maturity options, respectively. 
\bibliographystyle{chicago}
\bibliography{implied-sharpe}

\appendix

\section{Proofs}
\label{sec:proofs}
\begin{proof}[Proof of Proposition \ref{prop:approx psi}]
By applying Duhamel's principle in \eqref{eq:solve psi0}, we see that the zeroth order term $\psi_0$ is given as 
\eqstar{
\psi_0(t) =  - (\tfrac{1}{2}\lam^2)_0 (T-t) - \gamma \nu  p^{BS}(t),
}
where
\eqstar{
p^{BS}(t):=  \ee^x \Phi\left( \frac{1}{\sigma_0 \sqrt{T-t}} \bigl(x - k + \tfrac{1}{2}\sigma_0^2 (T-t) \bigr) \right) - \ee^k \Phi\left( \frac{1}{\sigma_0 \sqrt{T-t}} \bigl(x - k - \tfrac{1}{2}\sigma_0^2 (T-t) \bigr) \right).
}
$p^{BS}$ denotes the price of a European call option in a Black-Scholes model with volatility coefficient $\sigma_0.$ For the first order term, we have 
\eqstar{
0 = (\partial_t  + \Act_0) \psi_1 + \Act_1 \psi_0 -  (\tfrac{1}{2}\lam^2)_1, & &\psi_1(T,x,y;\nu)  = 0.
}
Again by applying Duhamel's principle, the results in Lemma \ref{lemma:lorig} and \eqref{eq:poly f}, we get 
\eqstar{
\psi_1(t) &= \int^T_t \dd t_1 \Pct_0(t,t_1)\left( \Act_1 \psi_0(t_1) -  (\tfrac{1}{2}\lam^2)_1\right)\\
&= \int^T_t \dd t_1 \Pct_0(t,t_1)\left( \Act_1 \Bigl(  -  (\tfrac{1}{2}\lam^2)_0 (T-t_1) - \gamma \nu  p^{BS}(t_1) \Bigr) -  (\tfrac{1}{2}\lam^2)_1\right)\\
&= -\int^T_t \dd t_1 \Pct_0(t,t_1) \Act_1\gamma \nu  p^{BS}(t_1) - \int^T_t \dd t_1 (\tfrac{1}{2}\lam^2)_1 \bigl( \Xct(t,t_1),\Yct(t,t_1) \bigr)\\
&=-\gamma \nu \int^T_t \dd t_1 \Gct_1(t,t_1) \Pct_0(t,t_1) p^{BS}(t_1) - \int^T_t \dd t_1 (\tfrac{1}{2}\lam^2)_1 \bigl( \Xct(t,t_1),\Yct(t,t_1) \bigr)\\
&= -\gamma \nu \int^T_t \dd t_1 \Gct_1(t,t_1) p^{BS}(t) - \int^T_t \dd t_1 (\tfrac{1}{2}\lam^2)_1 \bigl( \Xct(t,t_1),\Yct(t,t_1) \bigr).
}
In the above, we have also used the semi-group property of $\Pct$. For the second order term, we have 
\eqstar{
0 = (\partial_t  + \Act_0) \psi_2 +\Act_2 \psi_0 + \Act_1 \psi_1 -  (\tfrac{1}{2}\lam^2)_2 + (1-\rho^2) (\tfrac{1}{2}\beta^2)_0 (\partial_y \psi_1)^2, & &\psi_2(T,x,y;\nu)  = 0.
}
Once more, by Duhamel's principle, we obtain 
\eqlnostar{eq:main term}{
\psi_2(t) &= \int^T_t  \dd t_1\Pct_0(t,t_1)\Bigl(\Act_2 \psi_0(t_1) + \Act_1 \psi_1(t_1)\Bigr)\\
&+ \int^T_t \dd t_1\Pct_0(t,t_1) \Bigl(  -  (\tfrac{1}{2}\lam^2)_2 + (1-\rho^2) (\tfrac{1}{2}\beta^2)_0 (\partial_y \psi_1(t_1))^2\Bigr).
}
In the above we note the following while using the results in Lemma \ref{lemma:lorig} and the semi-group property of $\Pct_0:$
\eqstar{
\int^T_t  \dd t_1\Pct_0(t,t_1) \Act_2 \psi_0(t_1)  &= \int^T_t  \dd t_1\Pct_0(t,t_1) \Act_2 \Bigl(-  (\tfrac{1}{2}\lam^2)_0 (T-t) - \gamma \nu  p^{BS}(t_1) \Bigr)\\
&= - \gamma \nu \int^T_t  \dd t_1\Pct_0(t,t_1) \Act_2 p^{BS}(t_1)\\
&= - \gamma \nu \int^T_t  \dd t_1 \Gct_2(t,t_1) \Pct_0(t,t_1) \Pct_0(t_1,T) \varphi\\
&= - \gamma \nu \int^T_t  \dd t_1 \Gct_2(t,t_1) \Pct_0(t,T) \varphi\\
&= - \gamma \nu \int^T_t  \dd t_1 \Gct_2(t,t_1) p^{BS}(t).
}
Next, we have 
\eqlnostar{eq:term A}{
\int^T_t  \dd t_1\Pct_0(t,t_1)\Act_1 \psi_1(t_1) &= \int^T_t  \dd t_1\Pct_0(t,t_1)\Act_1 \Bigl( -\gamma \nu \int^T_{t_1} \dd t_2 \Gct_1(t_1,t_2) p^{BS}(t_1) - \int^T_{t_1} \dd t_2 (\tfrac{1}{2}\lam^2)_1 \bigl( \Xct(t_1,t_2),\Yct(t_1,t_2) \bigr)\Bigr).
}
First term can be computed as 
\eqstar{
-\gamma \nu   \int^T_t  \dd t_1\Pct_0(t,t_1)\Act_1  \int^T_{t_1} \dd t_2 \Gct_1(t_1,t_2) p^{BS}(t_1) &=  -\gamma \nu   \int^T_t  \dd t_1 \Gct_1(t,t_1) \Pct_0(t,t_1) \int^T_{t_1} \dd t_2 \Gct_1(t_1,t_2) \Pct_0(t_1,T) \varphi\\
&= -\gamma \nu   \int^T_t  \dd t_1 \Gct_1(t,t_1) \int^T_{t_1} \dd t_2 \Gct_1(t_1,t_2)\Pct_0(t,t_1) \Pct_0(t_1,T) \varphi\\
&= -\gamma \nu   \int^T_t  \dd t_1 \Gct_1(t,t_1) \int^T_{t_1} \dd t_2 \Gct_1(t_1,t_2)p^{BS}(t).
}
The second term in \eqref{eq:term A} can be resolved as 
\eqstar{
-\int^T_t  \dd t_1\Pct_0(t,t_1)\Act_1\int^T_{t_1} \dd t_2 (\tfrac{1}{2}\lam^2)_1 \bigl( \Xct(t_1,t_2),\Yct(t_1,t_2) \bigr)&= -\int^T_t  \dd t_1 \Gct_1(t,t_1) \Pct_0(t,t_1) \int^T_{t_1} \dd t_2 \Pct_0(t_1,t_2)(\tfrac{1}{2}\lam^2)_1\\
&= -\int^T_t  \dd t_1 \Gct_1(t,t_1) \int^T_{t_1} \dd t_2 \Pct_0(t,t_2)(\tfrac{1}{2}\lam^2)_1\\
&= -\int^T_t  \dd t_1 \Gct_1(t,t_1) \int^T_{t_1} \dd t_2(\tfrac{1}{2}\lam^2)_1  \bigl( \Xct(t,t_2),\Yct(t,t_2) \bigr).
}
The first of the two remaining terms in \eqref{eq:main term} can be computed using the property in \eqref{eq:poly f}:
\eqstar{
- \int^T_t \dd t_1\Pct_0(t,t_1) (\tfrac{1}{2}\lam^2)_2 = - \int^T_t \dd t_1 (\tfrac{1}{2}\lam^2)_2 (\Xct(t,t_1),\Yct(t,t_1)).
}
The second term is computed as 
\eqlnostar{eq:main aux term}{
&\int^T_t \dd t_1\Pct_0(t,t_1) (1-\rho^2) (\tfrac{1}{2}\beta^2)_0 (\partial_y \psi_1(t_1))^2\\
&= (1-\rho^2) (\tfrac{1}{2}\beta^2)_0 \int^T_t \dd t_1\Pct_0(t,t_1) \Bigl[ \gamma^2 \nu^2 \bigl( \partial_y \zeta_1(t_1)\bigr)^2 +\bigl( \partial_y \eta(t_1)\bigr)^2  - 2 \gamma \nu \partial_y \zeta_1(t_1) \partial_y \eta(t_1)\Bigr]
}
where we have defined $\zeta(t)$ and $\eta(t)$ as 
\eqstar{
&\zeta(t) :=  \int^T_t \dd t_1 \Gct_1(t,t_1) p^{BS}(t),& &
\eta(t) := - \int^T_t \dd t_1 (\tfrac{1}{2}\lam^2)_1 \bigl( \Xct(t,t_1),\Yct(t,t_1) \bigr).&
}
By straightforward computations, we have that 
\eqstar{
\partial_y \eta(t_1) = - \int^T_{t_1} \dd t_2 (\tfrac{1}{2}\lam^2)_{0,1}  =  - (\tfrac{1}{2}\lam^2)_{0,1} (T-t_1).
}
Next, we use the property in \eqref{eq:prop commute} to obtain the following:
\eqstar{
\int^T_t \dd t_1\Pct_0(t,t_1)  \partial_y \zeta_1(t_1) \partial_y \eta(t_1) &=- \int^T_t \dd t_1(\tfrac{1}{2}\lam^2)_{0,1}  (T-t_1) \Pct_0(t,t_1)  \partial_y \zeta(t_1)\\
&= - \int^T_t \dd t_1(\tfrac{1}{2}\lam^2)_{0,1}  (T-t_1) \partial_y \bigl( \Pct_0(t,t_1)  \int^T_{t_1} \dd t_2 \Gct_1(t_1,t_2)  \Pct_0(t_1,t_2) p^{BS}(t_2) \bigr)\\
&=- \int^T_t \dd t_1(\tfrac{1}{2}\lam^2)_{0,1}  (T-t_1) \partial_y \bigl( \Pct_0(t,t_1)  \int^T_{t_1} \dd t_2 \Pct_0(t_1,t_2) \Act_1 p^{BS}(t_2) \bigr)\\
&=- \int^T_t \dd t_1(\tfrac{1}{2}\lam^2)_{0,1}  (T-t_1) \partial_y \Bigl( \int^T_{t_1} \dd t_2 \Gct_1(t,t_2)   p^{BS}(t)\Bigr).
}
To tackle the final term in \eqref{eq:main aux term}, we note the following 
\eqstar{
 \partial_y \zeta(t_1) =  \partial_y \Bigl( \int^T_{t_1} \dd t_2 \Gct_1(t_1,t_2) \Bigr) p^{BS}(t_1) = (T-t_1) (\tfrac{1}{2}\sigma^2)_{0,1} (\partial^2_x - \partial_x ) p^{BS}(t_1).
}
Using the explicit expression for $p^{BS}(t),$ we obtain
\eqstar{
(\partial^2_x - \partial_x ) p^{BS}(t_1) = \frac{1}{\sigma_0 \sqrt{T-t_1}} \ee^{x} \phi\bigl(d_+(t_1) \bigr),
}
where $\phi$ is the density function of standard normal distribution. Then, from the calculations in Appendix B of \cite{lorig-4}, we obtain the following 
\eqstar{
& \int^T_t \dd t_1\Pct_0(t,t_1)  \bigl(\partial_y \zeta(t_1) \bigr)^2\\
&=  \frac{(\tfrac{1}{2}\sigma^2)^2_{0,1}}{2 \pi \sigma^2_0} \int^T_t \dd t_1 \frac{(T-t_1)^{3/2}}{\sqrt{T-t+t_1 -t}} \exp \left(2 k -  \frac{\bigl( (k-x) + \tfrac{1}{2}\sigma^2_0(T-t) \bigr)^2 }{\sigma^2_0 (T-t+t_1-t)}\right).
}
\end{proof}

\begin{proof}[Proof of Proposition \ref{prop:approx p}]
The linear operator $\Ach_0$ is the infinitesimal generator of a diffusion in $\Rb^2$ whose drift vector and covariance matrix are constant. Then, by Duhamel's principle, the zeroth order term $p_0,$ which is independent of $y,$ is given as 
\eqstar{
p_0(t) =  p^{BS}(t),
}
where for standard normal cdf $\Phi,$ we have 
\eqstar{
p^{BS}(t):=  \ee^x \Phi\left( \frac{1}{\sigma_0 \sqrt{T-t}} \bigl(x - k + \tfrac{1}{2}\sigma_0^2 (T-t) \bigr) \right) - \ee^k \Phi\left( \frac{1}{\sigma_0 \sqrt{T-t}} \bigl(x - k - \tfrac{1}{2}\sigma_0^2 (T-t) \bigr) \right).
}
For the first order term, we have the following equation
\eqstar{
(\partial_t  + \Ach_0) p_1 + \Ach_1 p_0 = 0,& & p_1(T,x,y) = 0.
}
Once again, by using Duhamel's principle, we obtain that 
\eqstar{
p_1(t)  =  \int^T_t \dd t_1 \Pch_0(t,t_1) \Ach_1 p_0(t_1).
}
%From Lemma \ref{lemma:lorig}, for any measurable function $f \in C^{n+2}(\Rb^2)$ whose partial derivatives of all orders less than or equal to $n+2$ are measurable and at most exponentially growing, we have that
%\eqstar{
%\Pct_0(t,t_1) \Act_n f = \Gc_n(t,t_1) \Pct_0(t,t_1)  f.
%}
Then, by applying the result in Lemma \ref{lemma:lorig}, we get
\eqstar{
p_1(t) &=  \int^T_t \dd t_1 \Gch_1(t,t_1)  \Pch_0(t,t_1) p_0(t_1)\\
&=  \int^T_t \dd t_1 \Gch_1(t,t_1)  \Pch_0(t,t_1)  \Pch_0(t_1,T) \varphi\\
&= \int^T_t \dd t_1 \Gch_1(t,t_1) p_0(t) ,
}
Now, for the second order term, we have 
\eqstar{
(\partial_t  + \Ach_0) p_2 + \Ach_1 p_1 +  \Ach_2 p_0 = 0,& & p_2(T,x,y) = 0.
}
By applying Duhamel's principle, we get 
\eqstar{
p_2(t) =  \int^T_t \dd t_1 \Pch_0(t,t_1) \bigl( \Ach_1 p_1(t_1) +  \Ach_2 p_0(t_1)  \bigr).
}
From the result in Lemma \ref{lemma:lorig} and the result for $p_1$, we obtain that 
\eqstar{
p_2(t) &= \int^T_t \dd t_1 \Pch_0(t,t_1)  \Ach_1 \int^T_{t_1} \dd t_2  \Pch_0(t_1,t_2) \Ach_1 p_0(t_2)\\
&+ \int^T_t \dd t_1\Gch_2(t,t_1) \Pch_0(t,t_1) p_0(t_1) \\
&= \int^T_t \dd t_1\Gch_1(t,t_1) \Pch_0(t,t_1) \int^T_{t_1} \dd t_2  \Gch_1(t_1,t_2)\Pch_0(t_1,t_2)\Pch_0(t_2,T)  \varphi\\
&+ \int^T_t \dd t_1 \Gch_2(t,t_1) \Pch_0(t,t_1) \Pch_0(t_1,T) \varphi\\
&= \left( \int^T_t \dd t_1 \Gch_2(t,t_1) +  \int^T_t \dd t_1\int^T_{t_1} \dd t_2 \Gch_1(t,t_1)\Gch_1(t_1,t_2)\right) p_0(t).
} 
\end{proof}

\section{Approximation formulas in Section \ref{sec:examples}}
\label{sec:formulas Heston}
In the first order correction  \eqref{eq:imp sharpe first}, we get
\eqstar{
&\int^T_t \dd t_1 (\tfrac{1}{2} \lam^2)_1 (\Xc(t,t_1), \Yc(t,t_1) )\\
&= \int^T_t \dd t_1\Bigl[ (\tfrac{1}{2} \lam^2)_{1,0} \bigl(x-\bar{x} - (\tfrac{1}{2} \sig^2)_0 (t_1-t)\bigr)+ (\tfrac{1}{2} \lam^2)_{0,1} \bigl(y-\bar{y} + (c-\rho\beta\lambda)_0 (t_1-t)\bigr)\Bigr]\\
&= (\tfrac{1}{2} \lam^2)_{1,0} \bigl((x-\bar{x})(T-t) - (\tfrac{1}{2} \sig^2)_0 \frac{1}{2}(T-t)^2\bigr) + (\tfrac{1}{2} \lam^2)_{0,1} \bigl((y-\bar{y})(T-t) + (c-\rho \beta \lambda)_0 \frac{1}{2}(T-t)^2\bigr).
}
Similarly, in the second order correction \eqref{eq:imp sharpe second}, we compute the following term
\eqstar{
&\int^T_t  \dd t_1 \Gc_1(t,t_1) \int^T_{t_1} \dd t_2(\tfrac{1}{2}\lam^2)_1  \bigl( \Xc(t,t_2),\Yc(t,t_2) \bigr)\\
&=-(\tfrac{1}{2}\lam^2)_{1,0} \frac{1}{2}(T-t)^2 \Biggl(\Bigl( (\tfrac{1}{2}\sig^2)_{1,0}(x-\bar{x}) + (\tfrac{1}{2}\sig^2)_{0,1}(y-\bar{y}) \Bigr)\\
&\qquad + \frac{1}{3}(T-t) \Bigl(  -(\tfrac{1}{2}\sig^2)_{1,0}  (\tfrac{1}{2}\sig^2)_0 +  (\tfrac{1}{2}\sig^2)_{0,1} (c-\rho\beta \lambda)_0 \Bigr) \Biggr)\\
&+(\tfrac{1}{2}\lam^2)_{0,1} \frac{1}{2}(T-t)^2 \Biggl(\Bigl( (c-\rho\beta\lambda)_{1,0}(x-\bar{x}) +(c-\rho\beta\lambda)_{0,1}(y-\bar{y}) \Bigr)\\
&\qquad + \frac{1}{3}(T-t) \Bigl(-(c-\rho\beta\lambda)_{1,0}  (\tfrac{1}{2}\sig^2)_0 +  (c-\rho\beta\lambda)_{0,1} (c-\rho\beta \lambda)_0  \Bigr) \Biggr).
}
\begin{itemize}
\item For the Heston model, we get 
\eqstar{
&\int^T_t \dd t_1 (\tfrac{1}{2} \lam^2)_2 (\Xc(t,t_1), \Yc(t,t_1) )\\
&=(\tfrac{1}{2}\lam^2)_{1,1} \Bigl( (T-t) (x-\bar{x}) (y-\bar{y}) + \frac{(T-t)^2}{2}\bigl((x - \bar{x})(c-\rho\beta\lambda)_0-  (y - \bar{y})(\tfrac{1}{2}\sig^2)_0 \bigr) - \frac{(T-t)^3}{3} (c-\rho\beta\lambda)_0 (\tfrac{1}{2}\sig^2)_0 \Bigr).
}
\item For the reciprocal Heston model, we get
\eqstar{
&\int^T_t \dd t_1 (\tfrac{1}{2}\lam^2)_2 (\Xc(t,t_1),\Yc(t,t_1))\\
&= (\tfrac{1}{2}\lam^2)_{0,2} \Bigl[(y - \bar{y})^2(T-t) + (T-t)^2(y - \bar{y}) (c-\rho \beta \lambda)_0 + \bigl((c-\rho \beta \lambda)_0\bigr)^2\frac{1}{3}(T-t)^3 \Bigr].
}
\end{itemize}
Finally, we have
\eqstar{
&\int^T_t \dd t_1(\tfrac{1}{2}\lam^2)_{0,1} (T-t_1)  \partial_y \Bigl( \int^T_{t_1} \dd t_2 \Gc_1(t,t_2) \Bigr)  p^{BS}(t)\\
&= \int^T_t \dd t_1(\tfrac{1}{2}\lam^2)_{0,1}  (T-t_1)^2  (\tfrac{1}{2}\sig^2)_{0,1} (\partial^2_x - \partial_x) p^{BS}(t)\\
&= \int^T_t \dd t_1(\tfrac{1}{2}\lam^2)_{0,1}  (T-t_1)^2   (\tfrac{1}{2}\sig^2)_{0,1} \frac{1}{\sigma_0 \sqrt{T-t}} \ee^x \phi\bigl( d_+(t,x;\sigma_0))\\
&= (\tfrac{1}{2}\lam^2)_{0,1}  (\tfrac{1}{2}\sig^2)_{0,1} \frac{1}{\sigma_0 \sqrt{T-t}} \ee^x \phi\bigl( d_+(t,x;\sigma_0)) \frac{(T-t)^3}{3}.
}

\end{document}